\documentclass{article}
\pdfoutput=1
\usepackage{etoolbox}
\newcommand*\samethanks[1][\value{footnote}]{\footnotemark[#1]}

\usepackage{booktabs}       

\newtoggle{arxivupload}
\togglefalse{arxivupload}
\newtoggle{neurips}
\togglefalse{neurips}

\usepackage{multirow,nicefrac}
\usepackage{makecell,upgreek}
\usepackage{footnote}
\usepackage{longtable}
\usepackage[shortlabels]{enumitem}
\usepackage{tablefootnote}
\usepackage[T1]{fontenc}
\usepackage{verbatim} 
\usepackage[utf8]{inputenc}
\usepackage{amsthm}
  \usepackage{boxedminipage}

\iftoggle{neurips}
{
  \usepackage{hyperref}

}
{
  \usepackage{fullpage}
\usepackage[breaklinks=true]{hyperref}

}

\usepackage{amsfonts,amssymb,mathrsfs}
\usepackage{mathtools}
\DeclareMathSymbol{\shortminus}{\mathbin}{AMSa}{"39}
\DeclareMathOperator*{\Expop}{\mathbb{E}}

\usepackage{prettyref}

\usepackage[vlined,boxed]{algorithm2e}

\usepackage[capitalise,nameinlink]{cleveref}
\Crefname{equation}{Eq.}{Eqs.}
\Crefname{assumption}{Assumption}{Assumptions}
\Crefname{condition}{Condition}{Conditions}

\usepackage{breakcites,bm}

\hypersetup{colorlinks,citecolor=black,linkcolor=blue,linktocpage=true}
\usepackage{latexsym}
\usepackage{relsize}
\usepackage{thm-restate}
\usepackage{appendix}

\usepackage{xcolor}
\usepackage{dsfont}

\newcommand{\defeq}{:=}

\numberwithin{equation}{section}

\def\grad{\nabla}

\DeclareFontFamily{U}{mathx}{\hyphenchar\font45}
\DeclareFontShape{U}{mathx}{m}{n}{
      <5> <6> <7> <8> <9> <10>
      <10.95> <12> <14.4> <17.28> <20.74> <24.88>
      mathx10
      }{}
\DeclareSymbolFont{mathx}{U}{mathx}{m}{n}
\DeclareFontSubstitution{U}{mathx}{m}{n}
\DeclareMathAccent{\widecheck}{0}{mathx}{"71}
\DeclareMathAccent{\wideparen}{0}{mathx}{"75}

\newcommand{\ignore}[1]{}

\newcommand{\R}{\mathbb{R}}

\newcommand{\opnorm}[1]{\|#1\|_{\op}}
\newcommand{\opnormm}[1]{\left\|#1\right\|_{\op}}
\newcommand{\prnn}[1]{\left(#1\right)}

\newcommand{\calN}{\mathcal{N}}

\newcommand{\matx}{\mathbf{x}}

\newcommand{\matu}{\mathbf{u}}
\newcommand{\matw}{\mathbf{w}}

\newcommand{\iidsim}{\overset{\mathrm{i.i.d}}{\sim}}

\newcommand{\dimx}{d_{x}}
\newcommand{\dimu}{d_{u}}

\newcommand{\mate}{\mathbf{e}}

\newcommand{\matz}{\mathbf{z}}

\newcommand{\op}{\mathrm{op}}

\newcommand{\fro}{\mathrm{F}}

	\theoremstyle{plain}
	\newtheorem{theorem}{Theorem}

	\newtheorem{lemma}{Lemma}[section]
	
	\newtheorem{fact}[lemma]{Fact}
	\newtheorem{corollary}{Corollary}[section]
	\newtheorem{proposition}[lemma]{Proposition}

	\theoremstyle{definition}

	\newtheorem{definition}{Definition}[section]

  \newtheorem{assumption}{Assumption}
  \newtheorem{condition}{Condition}[section]

\makeatletter
\newcommand{\neutralize}[1]{\expandafter\let\csname c@#1\endcsname\count@}
\makeatother

\newtheorem*{theorem*}{Theorem}
\newtheorem*{lemma*}{Lemma}
\newtheorem*{corollary*}{Corollary}
\newtheorem*{proposition*}{Proposition}
\newtheorem*{claim*}{Claim}
\newtheorem*{fact*}{Fact}
\newtheorem*{observation*}{Observation}

\newtheorem*{definition*}{Definition}
\newtheorem*{remark*}{Remark}
\newtheorem*{example*}{Example}

\newtheoremstyle{named}{}{}{\itshape}{}{\bfseries}{}{.5em}{\Cref{#3} {\normalfont (informal)} }
{}
\theoremstyle{named}

\theoremstyle{plain}

\DeclareMathAlphabet{\mathbfsf}{\encodingdefault}{\sfdefault}{bx}{n}

\DeclareMathOperator*{\argmin}{arg\,min}

\let\Pr\relax
\DeclareMathOperator{\Pr}{\mathbb{P}}

\newcommand{\norm}[1]{\|#1\|}
\newcommand{\ceil}[1]{\lceil #1 \rceil}

\newcommand{\E}{\mathbb{E}}

\newcommand{\trace}{\mathrm{tr}}

\newcommand{\poly}{\mathrm{poly}}

\renewcommand{\leq}{~\le~}
\renewcommand{\geq}{~\ge~}

\let\oldtfrac\tfrac
\renewcommand{\tfrac}[2]{\smash{\oldtfrac{#1}{#2}}}

\let\nablaold\nabla
\renewcommand{\nabla}{\nablaold\mkern-2.5mu}

\newcommand{\fronorm}[1]{\left\|#1\right\|_{\fro}}

\newcommand{\calH}{\mathcal{H}}

\newcommand{\diag}{\mathrm{diag}}

\newcommand{\calC}{\mathcal{C}}

\newcommand{\mb}[1]{\mathbf{#1}}
\newcommand{\matv}{\mb{v}}

\renewcommand{\P}[2]{P_{#1, #2}}
\newcommand{\traceb}[1]{\trace\left[ #1 \right]}

\newcommand{\Kst}{K_\star}
\newcommand{\Acl}{A_{\mathrm{cl}}}

\renewcommand{\epsilon}{\varepsilon}
\newcommand{\dlyap}[2]{\mathsf{dlyap}(#1, #2)}

\newcommand{\fnl}{f_{\mathrm{nl}}}
\newcommand{\flin}{G_{\mathrm{lin}}}
\newcommand{\Glin}{\flin}
\newcommand{\Jnl}{J_{\mathrm{nl}}}
\newcommand{\Jlin}{J_{\mathrm{lin}}}
\newcommand{\czero}{\betanl}
\newcommand{\cone}{\betanl}

\newcommand{\xtnl}{\matx_{t,\mathrm{nl}}}
\newcommand{\xtnlp}{\matx_{t+1,\mathrm{nl}}}
\newcommand{\xtnlm}{\matx_{t-1,\mathrm{nl}}}
\renewcommand{\nl}{\mathrm{nl}}

\newcommand{\fnorm}[1]{\| #1 \|_{\mathrm{F}}}

\newcommand{\Dk}{\mathrm{D}_K}
\newcommand{\eval}{\Big|}
\newcommand{\Pky}{P_{K,\gamma}}

\newcommand{\noisygrad}{\widetilde{\grad}}
\newcommand{\Ccost}{\calC_{\mathrm{cost}}}

\newcommand{\Pst}{P_\star}
\newcommand{\dxsphere}{\mathcal{S}^{\dimx -1}}
\newcommand{\fbar}{\overline{f}}

\newcommand{\Khat}{\widehat{K}}
\newcommand{\Gnl}{G_{\mathrm{nl}}}
\newcommand{\betanl}{\beta_{\mathrm{nl}}}
\newcommand{\rnl}{r_{\mathrm{nl}}}

\newcommand{\epsgrad}[1][\epsilon]{#1\text{-}\mathtt{Grad}}
\newcommand{\epseval}[1][\epsilon]{#1\text{-}\mathtt{Eval}}
\newcommand{\Phat}{\widehat{P}}
\newcommand{\Ajac}{A_{\mathrm{jac}}}
\newcommand{\Bjac}{B_{\mathrm{jac}}}
\newcommand{\calX}{\mathcal{X}}

\newcommand{\rst}{r_{\star}}

\newcommand{\mnl}{M_{\mathrm{nl}}}

\newcommand{\ts}{{T_s}} 

\newcommand{\nsmp}{N} 
\newcommand{\nstp}{M} 
\newcommand{\rlen}{H} 
\newcommand{\lrate}{\eta} 
\newcommand{\rroa}{{r_\text{roa}}} 

\newcommand{\hinf}{{\mathcal{H}_\infty}} 
\usepackage{natbib}

\title{Stabilizing Dynamical Systems via Policy Gradient Methods}
\author{Juan C. Perdomo\thanks{University of California, Berkeley, \texttt{jcperdomo@berkeley.edu}}   \and Jack Umenberger \thanks{MIT, \texttt{umnbrgr@mit.edu, msimchow@csail.mit.edu}}
 \and Max Simchowitz \samethanks[2] }

\date{\today}

\begin{document}

\maketitle
\begin{abstract}
Stabilizing an unknown control system is one of the most fundamental problems in control systems engineering.  In this paper, we provide a simple, model-free algorithm for stabilizing fully observed dynamical systems.  While model-free methods have become increasingly popular in practice due to their simplicity and flexibility, stabilization via direct policy search has received surprisingly little attention. Our algorithm proceeds by solving a series of discounted LQR problems, where the discount factor is gradually increased. We prove that this method efficiently recovers a stabilizing controller for linear systems, and for smooth, nonlinear systems within a neighborhood of their equilibria. Our approach overcomes a significant limitation of prior work, namely the need for a pre-given stabilizing control policy. 
We empirically evaluate the effectiveness of our approach on common control benchmarks.  
\end{abstract}

\setcounter{tocdepth}{0}


\section{Introduction}

Stabilizing an unknown control system is one of the most fundamental problems in control systems engineering. A wide variety of tasks - from maintaining a dynamical system around a desired equilibrium point, to tracking a reference signal (e.g a pilot's input to a plane) - can be recast in terms of stability. More generally, synthesizing an initial stabilizing controller is often a necessary first step towards solving more complex tasks, such as adaptive or robust control design \citep{sontag1999nonlinear, sontag2009}.

In this work, we consider the problem of finding a stabilizing controller for an unknown dynamical system via direct policy search methods. We introduce a simple procedure based off policy gradients which provably stabilizes a dynamical system around an equilibrium point. Our algorithm only requires access to a simulator which can return rollouts of the system under different control policies, and can efficiently stabilize both linear and smooth, nonlinear systems.

Relative to model-based approaches, model-free procedures, such as policy gradients, have two key advantages: they are conceptually simple to implement, and they are easily adaptable; that is, the same method can be applied in a wide variety of domains without much regard to the intricacies of the underlying dynamics. Due to their simplicity and flexibility, direct policy search methods have become increasingly popular amongst practitioners, especially in settings with complex, nonlinear dynamics which may be challenging to model. In particular, they have served as the main workhorse for recent breakthroughs in reinforcement learning and control \citep{silver2016mastering, mnih2015human, andrychowicz2020learning}.

Despite their popularity amongst practitioners, model-free approaches for continuous control have only recently started to receive attention from the theory community \citep{fazel2018global, kakade2020information, tu2019gap}. While these analyses have begun to map out the computational and statistical tradeoffs that emerge in choosing between model-based and model-free approaches, they all share a common assumption: that the unknown dynamical system in question is stable, or that an initial stabilizing controller is known. As such, they do not address the perhaps more basic question, \emph{how do we arrive at a stabilizing controller in the first place?}

\subsection{Contributions} 

We establish a reduction from stabilizing an unknown dynamical system to solving a series of discounted, infinite-horizon LQR problems via policy gradients, for which no knowledge of an initial stable controller is needed.  Our approach, which we call \emph{discount annealing}, gradually increases the discount factor and yields a control policy which is near optimal for the undiscounted LQR objective. To the best of our knowledge, our algorithm is the first model-free procedure shown to provably stabilize unknown dynamical systems, thereby solving an open problem from  \citet{fazel2018global}. 

We begin by studying linear, time-invariant dynamical systems with full state observation and assume  access to \emph{inexact} cost and gradient evaluations of the discounted, infinite-horizon LQR cost of a state-feedback controller $K$. Previous analyses (e.g., \citep{fazel2018global}) establish how such evaluations can be implemented with access to (finitely many, finite horizon) trajectories sampled from a simulator.  We show that our method recovers the controller $K_{\star}$ which is the optimal solution of the \emph{undiscounted} LQR problem in a bounded number of iterations, up to optimization and simulator error. The stability of the resulting $\Kst$ is guaranteed by known stability margin results for LQR. In short, we prove the following guarantee:
\newtheorem*{inf_theorem:linear_algorithm}{\Cref{theorem:linear_algorithm} (informal)} \begin{inf_theorem:linear_algorithm} 
For linear systems, discount annealing returns a stabilizing state-feedback controller which is also near-optimal for the LQR problem. It uses at most polynomially many, $\epsilon$-inexact gradient and cost evaluations, where the tolerance $\epsilon$ also depends polynomially on the relevant problem parameters. 
\end{inf_theorem:linear_algorithm}
Since both the number of queries and error tolerance are polynomial, discount annealing can be efficiently implemented using at most polynomially many samples from a simulator. 
 
Furthermore, our results extend to smooth, \emph{nonlinear} dynamical systems. Given access to a simulator that can return damped system rollouts, we show that our algorithm finds a controller that attains near-optimal LQR cost for the Jacobian linearization of the nonlinear dynamics at the equilibrium. We then show that this controller stabilizes the nonlinear system within a neighborhood of its equilibrium.

\newtheorem*{thm:nonlinear_informal}{\Cref{theorem:nonlinear_algorithm} (informal)}
\begin{thm:nonlinear_informal} 
Discount annealing returns a state-feedback controller which is exponentially stabilizing for smooth, nonlinear systems within a neighborhood of their equilibrium, using again only polynomially many samples drawn from a simulator.
\end{thm:nonlinear_informal}

In each case, the algorithm returns a near optimal solution $K$ to the relevant LQR problem (or local approximation thereof). Hence, the stability properties of $K$ are, in theory, no better than those of the optimal LQR controller $K_\star$. Importantly, the latter may have worse stability guarantees than the optimal solution of a corresponding robust control objective (e.g. $\mathcal{H}_{\infty}$ synthesis). Nevertheless, we focus on the LQR subroutine in the interest of simplicity, clarity, and in order to leverage prior analyses of model-free methods for LQR. Extending our procedure to robust-control objectives is an exciting direction for future work. 

Lastly, while our theoretical analysis only guarantees that the resulting controller will be stabilizing within a small neighborhood of the equilibrium, our simulations on nonlinear systems, such as the nonlinear cartpole, illustrate that discount annealing produces controllers that are competitive with established robust control procedures, such as $\mathcal{H}_{\infty}$ synthesis, without requiring any knowledge of the underlying dynamics.




\subsection{Related work}
Given its central importance to the field, stabilization of unknown and uncertain dynamical systems has received extensive attention within the controls literature. We review some of the relevant literature and point the reader towards classical texts for a more comprehensive treatment \citep{sontag2009,sastry2011adaptive,zhou1996robust, callier2012linear, zhou1998essentials}. 

\textbf{Model-based approaches.} Model-based methods construct approximate system models in order to synthesize stabilizing control policies. Traditional analyses consider stabilization of both linear and nonlinear dynamical systems in the asymptotic limit of sufficient data \citep{sastry2011adaptive,sastry2013nonlinear}. More recent, non-asymptotic studies  have focused almost entirely on \emph{linear} systems, where the controller is generated using data from multiple independent trajectories \citep{fiechter1997pac, dean2019sample, faradonbeh2018finite, faradonbeh2019randomized}.
Assuming the model is known, stabilizing policies may also be synthesized via convex optimization \citep{prajna2004nonlinear} by combining a `dual Lyapunov theorem' \citep{rantzer2001dual} with sum-of-squares programming \citep{parrilo2003semidefinite}. Relative to these analysis our focus is on strengthening the theoretical foundations of model-free procedures and establishing rigorous guarantees that policy gradient methods can also be used to generate stabilizing controllers.

\textbf{Online control.} Online control studies the problem of adaptively \emph{fine-tuning} the performance of an already-stabilizing control policy on a \emph{single} trajectory \citep{dean2018regret,faradonbeh2018optimality,cohen2019learning,mania2019certainty,simchowitz2020naive,hazan2020nonstochastic,simchowitz2020improper,kakade2020information}. Though early papers in this direction consider systems without pre-given stabilizing controllers \citep{abbasi2011regret}, their guarantees degrade exponentially in the system dimension (a penalty ultimately shown to be unavoidable by \cite{chen2021black}). Rather than fine-tuning an already stabilizing controller, we focus on the more basic problem of finding a controller which is stabilizing in the first place, and allow for the use of  multiple independent trajectories.

\textbf{Model-free approaches.} Model-free approaches eschew trying to approximate the underlying dynamics and instead directly search over the space of control policies. The landmark paper of \cite{fazel2018global} proves that, despite the non-convexity of the problem, direct policy search on the infinite-horizon LQR objective efficiently converges to the globally optimal policy, assuming the search is initialized at an already stabilizing controller. \cite{fazel2018global} pose the synthesis of this initial stabilizing controller via policy gradients as an open problem; one that we solve in this work. 

Following this result, there have been a large number of works studying policy gradients procedures in continuous control, see for example \cite{feng2020escaping, malik2019derivative, mohammadi2020linear, mohammadi2021convergence, zhang2021derivative} just to name a few. Relative to our analysis, these papers consider questions of policy finite-tuning, derivative-free methods, and robust (or distributed) control which are important, yet somewhat orthogonal to the stabilization question considered herein. The recent analysis by \cite{lamperski2020computing} is perhaps the most closely related piece of prior work. It  proposes a model-free, off-policy algorithm for computing a stabilizing controller for
deterministic LQR systems. Much like discount annealing, the algorithm also works by alternating between policy optimization (in their case by a closed-form policy improvement step based on the Riccati update) and increasing a damping factor. However, whereas we provide precise finite-time convergence guarantees to a stabilizing controller for both linear and nonlinear systems, the guarantees in \cite{lamperski2020computing} are entirely asymptotic and restricted to linear systems. Furthermore, we pay special attention to quantifying the various error tolerances in the gradient and cost queries to ensure that the algorithm can be efficiently implemented in finite samples.




\subsection{Background on stability of dynamical systems}
\label{sec:stability_defs}

Before introducing our results, we first review some of the basic concepts and definitions regarding stability of dynamical systems. In this paper, we study discrete-time, noiseless, time-invariant dynamical systems with states $\matx_t \in \R^{\dimx}$ and control inputs $\matu_t \in \R^{\dimu}$. In particular, given an initial state $\matx_0$, the dynamics evolves according to $\matx_{t+1} = G(\matx_t, \matu_t)$ where $G:\R^{\dimx}\times \R^{\dimu} \rightarrow \R^{\dimx}$ is a state transition map. An equilibrium point of a dynamical system is a state $\matx_\star \in \R^{\dimx}$ such that $G(\matx_\star, 0) = \matx_\star$. As per convention, we assume that the origin $\matx_\star = 0$ is the desired equilibrium point around which we wish to stabilize the system.

 This paper restricts its attention to static state-feedback policies of the form $\matu_t = K\matx_t$ for a fixed matrix $K \in \R^{\dimu \times \dimx}$. Abusing notation slightly, we conflate the matrix $K$ with its induced policy. Our aim is to find a policy $K$ which is \emph{exponentially stabilizing} around the  equilbrium point. 
 
 Time-invariant, linear systems, where $G(\matx,\matu) = A \matx + B \matu$ are stabilizable if and only if there exists a  $K$ such that $A+BK$ is a stable matrix \citep{callier2012linear}. That is if $\rho(A+BK) < 1$, where $\rho(X)$ denotes the spectral radius, or the largest eigenvalue magnitude, of a matrix $X$. For general nonlinear systems, our goal is to find controllers which satisfy the following general, quantitative definition of exponential stability (e.g Chapter 5.2 in \cite{sastry2013nonlinear}). Throughout, $\|\cdot\|$ denotes the Euclidean norm.
 \begin{definition} A controller $K$ is \emph{$(m,\alpha)$-exponentially stable} for dynamics $G$ if there exist constants $m, \alpha >0$ such that if inputs are chosen according to $\matu_t = K \matx_t$, the sequence of states $\matx_{t+1} = G(\matx_t,\matu_t)$ satisfy 
 \begin{align}
 \label{eq:exp_stable}
  \|\matx_{t}\| \le m \cdot \exp(- \alpha \cdot t) \|\matx_0\|.
 \end{align}
Likewise, $K$ is \emph{$(m,\alpha)$-exponentially stable on radius $r > 0$} if \eqref{eq:exp_stable} holds for all $\matx_0$ such that $\|\matx_0\| \le r$.
\end{definition}
For linear systems, a controller $K$ is stabilizing if and only if it is stable over the entire state space, however, the restriction to stabilization over a particular radius is in general needed for nonlinear systems. Our approach for stabilizing nonlinear systems relies on analyzing their \emph{Jacobian linearization} about the origin equilibrium. Given a continuously differentiable transition operator $G$, the local dynamics can be approximated by the Jacobian linearization $(A_{\mathrm{jac}},B_{\mathrm{jac}})$ of $G$ about the zero equilibrium; that is
\begin{align}
\label{eq:jacobian_linearization}
\Ajac \defeq \nabla_{\matx} G(\matx,\matu)\big{|}_{(\matx,\matu)=(0,0)}, \quad \Bjac \defeq \nabla_{\matu} G(\matx,\matu)\big{|}_{(\matx,\matu)=(0,0)}.
\end{align}
In particular, for $\matx$ and $\matu$ sufficiently small, $G(\matx,\matu) =\Ajac \matx + \Bjac \matu + f_{\mathrm{nl}}(\matx,\matu)$, where $f_{\mathrm{nl}}(\matx,\matu)$ is a nonlinear remainder from the Taylor expansion of $G$. To ensure stabilization via state-feedback is feasible, we assume throughout our presentation that the linearized dynamics $(\Ajac, \Bjac)$ are stabilizable. 



\section{Stabilizing Linear Dynamical Systems}
\label{sec:main_linear_results}
We now present our main results establishing how 
our algorithm, discount annealing, provably stabilizes linear dynamical systems via a reduction to direct policy search methods. We begin with the following preliminaries on the Linear Quadratic Regulator (LQR). 
\begin{definition}[LQR Objective]
\label{def:lqr_objective}
For a given starting state $\matx$, we define the LQR problem $\Jlin$ with discount factor $\gamma \in (0,1]$, dynamic matrices $(A,B)$, and state feedback controller $K$ as,
\begin{align*}
\Jlin(K \mid \matx, \gamma, A, B) \defeq  \sum_{t=0}^\infty \gamma^t \left( \matx_t^\top Q \matx_t + \matu_t^\top R \matu_t\right) 
\text{ s.t. }\,  \matu_t = K\matx_t,\,  \matx_{t+1} = A\matx_t + B\matu_t, \, \matx_0 = \matx.
\end{align*}
Here, $\matx_t \in \R^{\dimx}$, $\matu_t \in \R^{\dimu}$, and $Q, R$ are positive definite matrices. Slightly overloading notation, we define
\begin{align*}
	\Jlin(K \mid \gamma, A, B) \defeq  \Expop_{\matx_0 \sim \sqrt{\dimx} \cdot \dxsphere}[\Jlin(K \mid \matx, \gamma, A, B)],
\end{align*}
to be the same as the problem above, but where the initial state is now drawn from the uniform distribution over the sphere in $\R^{\dimx}$ of radius $\sqrt{\dimx}$.\footnote{This scaling is chosen so that the initial state distribution has identity covariance, and yields cost equivalent to $\matx_0 \sim \mathcal{N}(0,I)$.}
\end{definition}
To simplify our presentation, we adopt the shorthand $\Jlin(K \mid \gamma) \defeq \Jlin(K \mid \gamma, A, B)$ in cases where the system dynamics $(A,B)$ are understood from context. Furthermore, we assume that $(A,B)$ is stabilizable and that $\lambda_{\min}(Q),\lambda_{\min}(R) \ge 1$. It is a well-known fact that $K_{\star, \gamma} \defeq \argmin_K \Jlin(K \mid \gamma, A, B)$ achieves the minimum LQR cost over all possible control laws. We begin our analysis with the observation that the discounted LQR problem is equivalent to the undiscounted LQR problem with damped dynamics matrices.\footnote{This lemma is folklore within the controls community, see e.g. \cite{lamperski2020computing}.}
\begin{lemma}
\label{lemma:equivalence_lemma}
For all controllers $K$ such that $\Jlin(K\mid \gamma, A,B) < \infty$, 
\begin{align*}
\Jlin(K \mid \gamma, A, B) = \Jlin(K\mid 1, \sqrt{\gamma} \cdot A, \sqrt{\gamma} \cdot B).
\end{align*}
\end{lemma}
From this equivalence, it follows from basic facts about LQR that a controller $K$ satisfies $\Jlin(0 \mid \gamma, A, B) < \infty$ if and only if $\sqrt{\gamma}(A+BK)$ is stable. Consequently, for $\gamma < \rho(A)^{-2}$, the zero controller is stabilizing and one can solve the discounted LQR problem via direct policy search initialized at $K$ = 0 \citep{fazel2018global}. At this point, one may wonder whether the solution to this highly discounted problem yields a controller which stabilizes the undiscounted system. If this were true, running policy gradients (defined in \Cref{eq:policy_gradients_def}) to convergence, on a single discounted LQR problem, would suffice to find a stabilizing controller.
\begin{align}
\label{eq:policy_gradients_def}
	K_{t+1}  = K_t - \eta \grad_K \Jlin(K_t \mid \gamma),\quad   K_0 = 0, \quad \eta > 0
\end{align}
 Unfortunately, the following proposition shows that this is not the case.
\begin{proposition}[Impossibility of Reward Shaping]
\label{prop:reward_shaping}
Fix $A = \diag(0, 2)$. For any positive definite cost matrices $Q,R$ and discount factor $\gamma$ such that $\sqrt{\gamma}A$ is stable, there exists a matrix $B$ such that $(A,B)$ is controllable (and thus stabilizable), yet the optimal controller $K_{\star, \gamma} \defeq \argmin_{K} J(K \mid \gamma, A , B)$ on the discounted problem is such that $A+BK_{\star, \gamma}$ is unstable.
\end{proposition}

\begin{figure}[t!]
\setlength{\fboxsep}{2mm}
\begin{boxedminipage}{\textwidth}
\begin{center}
\underline{Discount Annealing} \\
\end{center}
{\bf Initialize:} Objective $J(\cdot \mid \cdot)$, $\gamma_0 \in (0, \rho(A)^{-2})$,\; $K_0 \leftarrow 0$,\; and  $Q \leftarrow I$, $R \leftarrow I$\\
\\
{\bf For $t = 0, 1, \dots$}
\begin{enumerate}
	\item If $\gamma_t = 1$, run policy gradients once more as in Step 2, break, and return the resulting $K'$.
	\item Using policy gradients (see \Cref{eq:policy_gradients_def}) initialized at $K_{t}$, find $K'$ such that:
	\begin{align}
	\label{eq:approx_discount_annealing}
		\Jlin(K' \mid \gamma_t) - \min_{K} \Jlin(K \mid \gamma_t) \leq \dimx.
	\end{align}
	\item Update initial controller $K_{t+1} \leftarrow K'$.

	\item Using binary or random search, find a discount factor $\gamma' \in [\gamma_t, 1]$ such that  
	\begin{align}
	\label{eq:discount_factor_search}
		2.5 J(K_{t+1} \mid \gamma_t) \leq J(K_{t+1} \mid \gamma') \leq 8 J(K_{t+1} \mid \gamma_t).
	\end{align}
	\item Update the discount factor $\gamma_{t+1} \leftarrow \gamma'$.
\end{enumerate}
\end{boxedminipage}
\caption{Discount annealing algorithm. The procedure is identical for both linear and nonlinear systems. For linear, we initialize $J = \Jlin(\cdot \mid \gamma_0)$ and for nonlinear $J = \Jnl(\cdot \mid \gamma_0, r_\star)$ where $r_\star$ is chosen as in \Cref{theorem:nonlinear_algorithm}. See \Cref{theorem:linear_algorithm}, \Cref{theorem:nonlinear_algorithm}, and \Cref{section:finite_samples} for details regarding policy gradients and binary (or random) search. The  constants above are chosen for convenience, any constants $c_1, c_2$ such that $1 < c_1 < c_2$ suffice.}
\label{fig:discount_annealing}
\end{figure}

We now describe the discount annealing procedure for linear systems (\Cref{fig:discount_annealing}), which provably recovers a stabilizing controller $K$. For simplicity, we present the algorithm assuming access to noisy, bounded cost and gradient evaluations which satisfy the following definition. Employing standard arguments from \citep{fazel2018global, flaxman2005}, we illustrate how these evaluations can be efficiently implemented using polynomially many samples drawn from a simulator in \Cref{section:finite_samples}.

\begin{definition}[Gradient and Cost Queries]\label{defn:queries} Given an error parameter $\epsilon>0$ and a function $J:\R^{d} \rightarrow \R$,  $\epsgrad(J, \matz)$ returns a vector $\noisygrad$ such that $\fnorm{\noisygrad - \grad J(\matz)} \leq \epsilon$. Similarly, $\epseval(J, \matz, c)$ returns a scalar $v$ such that $|v - \min\{J(\matz), c\}| \leq \epsilon$.
\end{definition}

 The procedure leverages the equivalence (\Cref{lemma:equivalence_lemma}) between discounted costs and damped dynamics for LQR, and the consequence that the zero controller  is stabilizing if we choose $\gamma_0$ sufficiently small. Hence, for this discount factor, we may apply policy gradients initialized at the zero controller in order to recover a controller $K_1$ which is near-optimal for the $\gamma_0$ discounted objective. 

Our key insight is that, due to known stability margins for LQR controllers, $K_1$ is stabilizing for the $\gamma_1$ discounted dynamics for some discount factor $\gamma_1 > (1+c) \gamma_0$, where $c$ is a small constant that has a uniform lower bound. Therefore, $K_1$ has finite cost on the $\gamma_1$ discounted problem, so that we may again use policy gradients initialized at $K_1$ to compute a near-optimal controller $K_2$ for this larger discount factor. By iterating, we have that $\gamma_t \geq (1+ c)^t \gamma_0$ and can increase the discount factor up to 1, yielding a near-optimal stabilizing controller for the undiscounted LQR objective. 

The rate at which we can increase the discount factors $\gamma_t$ depends on certain properties of the (unknown) dynamical system. Therefore, we opt for binary search to compute the desired $\gamma$ in the absence of system knowledge. This yields the following guarantee, which we state in terms of properties of the matrix $\Pst$, the optimal value function for the undiscounted LQR problem, which satisfies $\min_K \Jlin(K \mid 1) =\traceb{\Pst}$ (see \Cref{sec:linear_appendix} for further details).

\begin{theorem}[Linear Systems]
\label{theorem:linear_algorithm}
Let $M_{\mathrm{lin}} \defeq \max\{ 16 \traceb{P_\star}, \Jlin(K_0 \mid \gamma_0)\}$. The following statements are true regarding the discount annealing algorithm when run on linear dynamical systems:
\begin{enumerate}[a)]
	\item Discount annealing returns a controller $\Khat$ which is $(\sqrt{2\trace[\Pst]}, (4 \traceb{\Pst})^{-1})$-exponentially stable.
	\item If $\gamma_0 < 1$, the algorithm is guaranteed to halt whenever $t$ is greater than $64 \traceb{\Pst}^4 \log(1 / \gamma_0)$.
\end{enumerate}
Furthermore, at each iteration $t$:
\begin{enumerate}[c)]
	\item Policy gradients as defined in \Cref{eq:policy_gradients_def} achieves the guarantee in \Cref{eq:approx_discount_annealing} using only $\poly(M_{\mathrm{lin}}, \opnorm{A} , \opnorm{B})$ many queries to $\epsgrad(\cdot , \Jlin(\cdot \mid \gamma))$ as long as $\epsilon$ is less than $ \poly(M_{\mathrm{lin}}^{-1}, \opnorm{A}^{-1}, \opnorm{B}^{-1})$. 
	
	\item The noisy binary search algorithm (see \Cref{fig:noisy_binary_search}) returns a discount factor $\gamma'$ satisfying \Cref{eq:discount_factor_search} using at most $\ceil{4\log( \traceb{\Pst})} + 10$ many queries to $\epseval(\cdot, \Jlin(\cdot \mid \gamma))$ for $\epsilon = .1\dimx$.
\end{enumerate}
\end{theorem}
We remark that since $\epsilon$ need only be polynomially small in the relevant problem parameters, each call to $\epsgrad$ and $\epseval$ can be carried out using only polynomially many samples from a  simulator which returns finite horizon system trajectories under various control policies. We make this claim formal in \Cref{section:finite_samples}.

\begin{proof}
We prove part $b)$ of the theorem and defer the proofs of the remaining parts of  to \Cref{sec:linear_appendix}. Define $\Pky$ to be the solution to the discrete-time Lyapunov equation. That is for $\sqrt{\gamma}(A+BK)$ stable, $\Pky$ solves:
\begin{align}
\label{eq:def_pky}
	\Pky \defeq Q + K^\top R K + \gamma(A+BK)^\top \Pky (A + BK).
\end{align}
Using this notation, $\Pst = P_{\Kst,1}$ is the solution to the above Lyapunov equation with $\gamma=1$. The key step of the proof is \Cref{lemma:lower_bound_gamma}, which uses  Lyapunov theory to verify the following: given the current discount factor $\gamma_t$, an idealized discount factor $\gamma_{t+1}'$ defined by 
\begin{align*}
	\gamma_{t+1}' \defeq  \left(\frac{1}{8 \opnorm{P_{K_{t+1}, \gamma_{t}}}^4}  + 1 \right)^2 \gamma_t,
\end{align*}
satisfies $\Jlin(K_{t+1} \mid \gamma_{t+1}') = \trace[P_{K_{t+1}, \gamma_{t+1}'}] \leq 2\traceb{P_{K_{t+1}, \gamma_{t}}} = 2 \Jlin(K_{t+1} \mid \gamma_t)$. Since the control cost is non-decreasing in $\gamma$, the binary search update in Step 4 ensures that the actual $\gamma_{t+1}$ also satisfies
\begin{align*}
	\gamma_{t+1} \geq \left(\frac{1}{8 \opnorm{P_{K_{t+1}, \gamma_{t}}}^4}  + 1 \right)^2 \gamma_t \geq \left(\frac{1}{ 128 \traceb{\Pst}^4}  + 1 \right)^2 \gamma_t,  
\end{align*}
The following calculation (which uses $\dimx \le \traceb{\Pst}$ for $\lambda_{\min}(Q) \ge 1$) justifies the second inequality above:
\begin{align*}
\opnorm{P_{K_{t+1}, \gamma_{t}}}^4 \leq \traceb{P_{K_{t+1}, \gamma_{t}}}^4 = \Jlin(K_{t+1} \mid \gamma_t)^4 \leq (\min_K \Jlin(K\mid \gamma_t) + \dimx)^4 \leq 16 \traceb{\Pst}^4.
\end{align*}
Therefore, $\gamma_t \geq ( 1 / (128 \traceb{\Pst}^4)  + 1 )^{2t} \gamma_0$. The precise bound follows from taking logs of both sides and using the numerical inequality $\log(1+x)\leq x$ to simplify the denominator.

\end{proof}


\section{Stabilizing Nonlinear Dynamical Systems}

We now extend the guarantees of the discount annealing algorithm to smooth, nonlinear systems. Whereas our study of linear systems explicitly leveraged the equivalence of discounted costs and damped dynamics, our analysis for nonlinear systems \emph{requires} access to system rollouts under damped dynamics, since the previous equivalence between discounting and damping breaks down in nonlinear settings. 

More specifically, in this section, we assume access to a simulator which given a controller $K$, returns trajectories generated according to $\matx_{t+1} = \sqrt{\gamma}\Gnl(\matx_t, K \matx_{t})$ for any damping factor $\gamma \in (0,1]$, where $\Gnl$ is the transition operator for the nonlinear system. While such trajectories may be infeasible to generate on a physical system, we believe these are reasonable to consider when dynamics are represented using software simulators, as is often the case in practice \citep{lewis2003robot, sim2real}. 

The discount annealing algorithm for nonlinear systems is almost identical to the algorithm for linear systems. It again works by repeatedly solving a series of quadratic cost objectives on the nonlinear dynamics as defined below, and progressively increasing the damping factor $\gamma$.

 \begin{definition}[Nonlinear Objective]
\label{def:nonlinear_objective}
For a statefeedback controller $K: \R^{\dimx} \rightarrow \R^{\dimu}$, damping factor $\gamma \in (0,1]$, and an initial state $\matx$, we define:
\begin{align}
&\Jnl(K  \mid \matx, \gamma) \defeq  \sum_{t=0}^\infty  \matx_t^\top Q \matx_t + \matu_t^\top R \matu_t \label{eq:nonlinear_cost} \\
&\text{ s.t }  \matu_t = K\matx_t,\quad   \matx_{t+1} = \sqrt{\gamma} \cdot \Gnl( \matx_t, \matu_t), \quad \matx_0 = \matx. \label{eq:nonlinear_dynamics}
\end{align}
Overloading notation as before, we let $\Jnl(K \mid \gamma, r)  \defeq \E_{\matx \sim r \cdot \dxsphere} \left[ \Jnl(K \mid \gamma,  \matx) \right] \times \frac{\dimx}{r^2}$.
\end{definition}
The normalization by $\dimx/r^2$ above is chosen so that the nonlinear objective coincides with the LQR objective when $\Gnl$ is in fact linear. Relative to the linear case, the only algorithmic difference for nonlinear systems is that we introduce an extra parameter $r$ which determines the radius for the initial state distribution. As established in \Cref{theorem:nonlinear_algorithm}, this parameter must be chosen small enough to ensure that discount annealing succeeds. Our analysis pertains to dynamics which satisfy the following smoothness definition.

\begin{assumption}[Local Smoothness]
\label{ass:smooth_dynamics}
The transition map $\Gnl$ is continuously differentiable. Furthermore, there exist  $\rnl, \betanl > 0$ such that for all $(\matx,\matu) \in \R^{\dimx+\dimu}$ with $\norm{\matx} + \norm{\matu} \le \rnl$,  
\begin{align*}
	\opnorm{\grad_{\matx,\matu}\Gnl(\matx, \matu) - \grad_{\matx, \matu} \Gnl(0,0)} \leq \betanl ( \norm{\matx} + \norm{\matu}).
\end{align*}
\end{assumption}
For simplicity, we assume $\betanl \geq 1$ and $\rnl \leq 1$. Using \Cref{ass:smooth_dynamics}, we can apply Taylor's theorem to rewrite $\Gnl$ as its Jacobian linearization around the equilibrium point, plus a nonlinear remainder term.
\begin{lemma}
\label{lemma:jacobian}
If $\Gnl$ satisfies \Cref{ass:smooth_dynamics}, then all $\matx, \matu$ for which $\|\matx\| + \|\matu\| \leq \rnl$, 
\begin{align}
\label{eq:nonlinear_assumption}
\Gnl(\matx,\matu) = \Ajac \matx + \Bjac \matu + \fnl(\matx,\matu),
\end{align}
where $\norm{\fnl(\matx,\matu)} \leq \betanl (\|\matx\|^2+\|\matu\|^2)$, $\|\nabla \fnl(\matx,\matu)\| \leq \betanl (\|\matx\|+\|\matu\|)$, and where $(\Ajac,\Bjac)$ are the system's Jacobian linearization matrices defined in \Cref{eq:jacobian_linearization}. 
\end{lemma} 

Rather than trying to directly understand the behavior of stabilization procedures on the nonlinear system, the key insight of our nonlinear analysis is that we can reason about the performance of a state-feedback controller on the nonlinear system via its behavior on the system's Jacobian linearization. In particular, the following lemma establishes how any controller which achieves finite discounted LQR cost for the Jacobian linearization is guaranteed to be exponentially stabilizing on the damped nonlinear system for initial states that are small enough. Throughout the remainder of this section, we define $\Jlin(\cdot \mid \gamma) \defeq \Jlin(\cdot \mid \gamma, \Ajac, \Bjac)$ as the LQR objective from \Cref{def:lqr_objective} where $(A,B) = (\Ajac,\Bjac)$. 
\begin{lemma}[Restatement of \Cref{lemma:exp_convergence}]\label{lem:stable_for_linear_to_nl} Suppose that $\calC_{J} = \Jlin(K \mid \gamma) < \infty$, then $K$ is $(\calC_{J}^{1/2}, (4\calC_{J})^{-1})$ exponentially stable on the damped system $\sqrt{\gamma}\Gnl$ over  radius $r = \rnl \; / \;(\betanl \calC_{J}^{3/2})$.
\end{lemma}

The second main building block of our nonlinear analysis is the observation that if the dynamics are locally smooth around the equilibrium point, then by \Cref{lemma:jacobian}, decreasing the radius $r$ of the initial state distribution $\matx_0 \sim r \cdot \dxsphere$ reduces the magnitude of the nonlinear remainder term $\fnl$. Hence, the nonlinear system smoothly approximates its Jacobian linearization. More precisely, we establish that the difference in gradients and costs between $\Jnl(K \mid \gamma, r)$ and $\Jlin(K \mid \gamma)$ decrease linearly with the radius $r$.
\begin{proposition}
\label{prop:diff_lin_nl}
Assume $\Jlin(K \mid \gamma) < \infty$. Then, for $\Pky$ defined as in \Cref{eq:def_pky}:
\begin{enumerate}[a)]
	\item If $r \leq \frac{\rnl}{2 \czero \opnorm{\Pky}^2}$, then 	$\big| \Jnl(K \mid \gamma, r) - \Jlin(K \mid \gamma)\big| \leq 8  \dimx \czero \opnorm{\Pky}^4 \cdot r.$

	\item If $r \leq \frac{1}{12 \betanl \opnorm{\Pky}^{5/2}}$, then, $\fnorm{\grad_K \Jnl(K \mid \gamma, r) - \grad_K \Jlin(K \mid \gamma)} \leq 48 \dimx  \czero (1 + \opnorm{B})  \opnorm{\Pky}^{7} \cdot r$
\end{enumerate}
\end{proposition}

Lastly, because policy gradients on linear dynamical systems is robust to inexact gradient queries, we show that for $r$ sufficiently small, running policy gradients on $\Jnl$ converges to a controller which has performance close to the optimal controller for the LQR problem with dynamic matrices $(\Ajac, \Bjac)$. As noted previously, we can then use \Cref{lem:stable_for_linear_to_nl} to translate the performance of the optimal LQR controller for the Jacobian linearization to an exponential stability guarantee for the nonlinear dynamics. Using these insights, we establish the following  theorem regarding discount annealing for nonlinear dynamics.
\begin{theorem}[Nonlinear Systems] 
\label{theorem:nonlinear_algorithm}
Let $\mnl \defeq \max \{ 21\traceb{\Pst}, \Jlin(K_0 \mid \gamma_0) \}$. The following statements are true regarding the discount annealing algorithm for nonlinear dynamical systems when $r_\star$ is less than a fixed quantity that is $\poly(1/ \mnl, 1 / \opnorm{A} , 1 / \opnorm{B}, \rnl / \betanl)$
	\begin{enumerate}[a)]
\item Discount annealing returns a controller $\Khat$ which is $(\sqrt{2\trace[\Pst]}, (8 \traceb{\Pst})^{-1})$-exponentially stable over a radius $r = \rnl \;/ \;(8 \betanl \traceb{\Pst}^2)$
	\item If $\gamma_0 < 1$, the algorithm is guaranteed to halt whenever $t$ is greater than $ 64 \traceb{\Pst}^4 \log(1 / \gamma_0)$.
\end{enumerate}
Furthermore, at each iteration $t$:
\begin{enumerate}[c)]
	\item Policy gradients achieves the guarantee in \Cref{eq:approx_discount_annealing} using only $\poly(M_{\mathrm{nl}}, \opnorm{A} , \opnorm{B})$ many queries to $\epsgrad(\cdot , \Jnl(\cdot \mid \gamma))$ as long as $\epsilon$ is less than some fixed polynomial $ \poly(M_{\mathrm{nl}}^{-1}, \opnorm{A}^{-1}, \opnorm{B}^{-1})$.  
	\item Let $c_0$ denote a universal constant. With probability $1-\delta$, the noisy random search algorithm (see \Cref{fig:noisy_binary_search}) returns a discount factor $\gamma'$ satisfying \Cref{eq:discount_factor_search} using at most $c_0 \cdot \traceb{\Pst}^4 \log(1 / \delta)$ queries to $\epseval(\cdot, \Jnl(\cdot \mid \gamma, \rst))$ for $\epsilon = .01 \dimx$.
	\end{enumerate}
\end{theorem}
We note that while our theorem only guarantees that the controller is stabilizing around a polynomially small neighborhood of the equilibrium, in  experiments, we find that the resulting controller successfully stabilizes the dynamics for a wide range of initial conditions. Relative to the case of linear systems where we leveraged the monotonicity of the LQR cost to search for discount factors using binary search, this monotonicity breaks down in the case of nonlinear systems and we instead analyze a random search algorithm to simplify the analysis.



\section{Experiments}

In this section, we evaluate the ability of the discount annealing algorithm to stabilize a simulated nonlinear system. Specifically, we consider the familiar cart-pole, with $\dimx=4$ (positions and velocities of the cart and pole), and $\dimu=1$ (horizontal force applied to the cart). 
The goal is to stabilize the system with the pole in the unstable `upright' equilibrium position.
For further details, including the precise dynamics, see \Cref{sec:cartpole}.
The system was simulated in discrete-time with a simple forward Euler discretization, i.e., $\matx_{t+1} = \matx_t + \ts \dot{\matx}_t$, where $\dot{\matx}_t$ is given by the continuous time dynamics, and $\ts=0.05$ (20Hz).
Simulations were carried out in PyTorch \citep{paszke2019pytorch} and run on a single GPU. 

\textbf{Setup}. The discounted annealing algorithm of \Cref{fig:discount_annealing} was implemented as follows.
In place of the true infinite horizon discounted cost $\Jnl(K\mid\gamma, r)$ in \Cref{eq:nonlinear_cost} we use a finite horizon, finite sample Monte Carlo approximation as described in \Cref{section:finite_samples}, 
\begin{align*}
\Jnl(K \mid \gamma) \approx \frac{1}{\nsmp}{\sum}_{i=1}^\nsmp \Jnl^{(H)}(K \mid \gamma, \matx^{(i)}), \quad \matx^{(i)} \sim r \cdot \dxsphere.
\end{align*}
Here, $\Jnl^{(H)}(K \mid \gamma, \matx) = \sum_{j=0}^{H-1}  \matx_t^\top Q \matx_t + \matu_t^\top R \matu_t$, is the length $H$, finite horizon cost of a controller $K$ in which the states evolve according to the $\sqrt{\gamma}$ damped dynamics from \Cref{eq:nonlinear_dynamics} and $\matu_t = K\matx_t$. 
We used $\nsmp=5000$ and $H=1000$ in our experiments.
For the cost function, we used $Q=\ts \cdot  I$ and $R=\ts$. We compute unbiased approximations of the gradients using automatic differentiation on the finite horizon objective $\Jnl^{(H)}$.
\begin{table}
	\caption{Final region of attraction radius $\rroa$ as a function of the initial state radius $r$ used during training (discount annealing). 
	We report the 
 [\texttt{min}, \texttt{max}]
	values of $\rroa$ over 5 independent trials. 	
	The optimal LQR policy for the linearized system achieved $\rroa=0.703$ when applied to the nonlinear system.
	We also synthesized an $\mathcal{H}_\infty$ optimal controller for the linearized dynamics, which achieved $\rroa=0.506$. }
	\label{tab:roas}
\vspace{10pt}
\centering
	\begin{tabular}{lccccc}
	\toprule
	$r$     & 0.1 & 0.3 & 0.5 & 0.6 & 0.7   \\
	\midrule
	$\rroa$  &  [0.702, 0.704]  & [0.711, 0.713] & [0.727, 0.734] & [0.731, 0.744] & [0.769, 0.777]      \\
	\bottomrule
\end{tabular}
\end{table}

Instead of using SGD updates for policy gradients, we use Adam \citep{kingma2014adam} with a learning rate of $\lrate=0.01/r$. Furthermore,
we replace the policy gradient termination criteria in Step 2 (\Cref{eq:approx_discount_annealing}) by instead halting after a fixed number $(\nstp=200)$ of gradient descent steps. 
We wish to emphasize that the hyperparameters $(\nsmp,H,\lrate,\nstp)$ were not optimized for performance.
In particular, for $r=0.1$, we found that as few as $\nstp=40$ iterations of policy gradient and horizons as short as $\rlen=400$ were sufficient.
Finally, we used an initial discount factor $\gamma_0=0.9 \cdot \| \Ajac \|_2^{-2}$, where $\Ajac$ denotes the linearization of the (discrete-time) cart-pole about the vertical equilibrium.

\textbf{Results}. We now proceed to discuss the performance of the algorithm, focusing on three main properties of interest:
i) the number of iterations of discount annealing required to find a stabilizing controller (that is, increase $\gamma_t$ to 1),
ii) the maximum radius $r$ of the ball of initial conditions $\matx_0 \sim r \cdot \dxsphere$ for which discount annealing succeeds at stabilizing the system, and
iii) the radius $\rroa$ of the largest ball contained within the region of attraction (ROA) for the policy returned by discount annealing. 
Although the true ROA (the set of all initial conditions such that the closed-loop system converges asymptotically to the equilibrium point) is not necessarily shaped like a ball 
(as the system is more sensitive to perturbations in the position and velocity of the pole than the cart),
we use the term region of attraction radius to refer to the radius of the largest ball contained in the ROA. 

Concerning (i), discount annealing reliably returned a stabilizing policy in less than 9 iterations. 
Specifically, over 5 independent trials for each initial radius $r \in \lbrace 0.1,0.3,0.5,0.7\rbrace$ (giving 20 independent trials, in total) the algorithm never required more than 9 iterations to return a stabilizing policy. 

Concerning (ii), discount annealing reliably stabilized the system for $r\leq 0.7$. 
For $r\approx 0.75$, we observed trials in which the state of the damped system ($\gamma<1$) diverged to infinity. 
For such a rollout, the gradient of the cost is not well-defined, and policy gradient is unable to improve the policy, which prevents discount annealing from finding a stabilizing policy. 

Concerning (iii), in \Cref{tab:roas} we report the final radius $\rroa$ for the region of attraction of the final controller returned by discount annealing as a function of the training radius $r$.
We make the following observations. 
Foremost, the policy returned by discount annealing extends the radius of the ROA beyond the radius used during training, i.e. $\rroa>r$.
Moreover, for each $r > .1$, the $\rroa$ achieved by discount annealing is greater than the $\rroa=0.703$ achieved by the exact optimal LQR controller \emph{and} the $\rroa=0.506$ achieved by the exact optimal $\hinf$ controller for the system's Jacobian linearization (see \Cref{tab:roas}). 
(The $\hinf$ optimal controller mitigates the effect of worst-case additive state disturbances on the cost; cf. \Cref{sec:hinf} for details).

One may hypothesize that this is due to the fact that discount annealing directly operates on the true nonlinear dynamics whereas the other baselines (LQR and $\calH_{\infty}$ control), find the optimal controller for an idealized linearization of the dynamics. Indeed, there is evidence to support this hypothesis.
In \Cref{fig:err_vs_discount} presented in \Cref{sec:experiments_app},
we plot the error $\| K_\text{pg}^\star(\gamma_t) - K_\text{lin}^\star(\gamma_t)\|_F$ between the policy $K_\text{pg}^\star(\gamma_t)$ returned by policy gradients, and the optimal LQR policy $K_\text{lin}^\star(\gamma_t)$ for the (damped) linearized system, as a function of the discount factor $\gamma_t$ used in each iteration of the discount annealing algorithm. 
For small training radius, such as $r=0.05$, $K_\text{pg}^\star(\gamma_t) \approx K_\text{lin}^\star(\gamma_t)$ for all $\gamma_t$.
However, for larger radii (i.e $r=0.7$), we see that $\| K_\text{pg}^\star(\gamma_t) - K_\text{lin}^\star(\gamma_t)\|_F$ steadily increases as $\gamma_t$ increases.

That is, as discount annealing increases the discount factor $\gamma$ and the closed-loop trajectories explore regions of the state space where the dynamics are increasingly nonlinear, $K_\text{pg}^\star$ begins to diverge from $K_\text{lin}^\star$. 
Moreover, at the conclusion of discount annealing $K_\text{pg}^\star(1)$ achieves 
a lower cost, namely 
[15.2, 15.4] vs [16.5, 16.8]  (here $[a,b]$ denotes [$\min$, $\max$] over 5 trials) 
and larger $\rroa$, namely [0.769, 0.777] vs [0.702, 0.703], 
than $K_\text{lin}^\star(1)$, 
suggesting that the method has indeed adapted to the nonlinearity of the system. Similar observations as to the behavior of controllers fine tuned via policy gradient methods are predicted by the theoretical results from \cite{qu2020combining}.


\section{Discussion}

This works illustrates how one can provably stabilize a broad class of dynamical systems via a simple model-free procedure based off policy gradients. In line with the simplicity and flexibility that have made model-free methods so popular in practice, our algorithm works under relatively weak assumptions and with little knowledge of the underlying dynamics. Furthermore, we solve an open problem from previous work \citep{fazel2018global} and take a step towards placing model-free methods on more solid theoretical footing. We believe that our results raise a number of interesting questions and directions for future work. 

In particular, our theoretical analysis states that discount annealing returns a controller whose stability properties are similar to those of the optimal LQR controller for the system's Jacobian linearization. We were therefore quite surprised when in experiments, the resulting controller had a significantly better radius of attraction than the exact optimal LQR and $\mathcal{H}_{\infty}$ controllers for the linearization of the dynamics. It is an interesting and important direction for future work to gain a better understanding of exactly when and how model-free procedures are adaptive to the nonlinearities of the system and improve upon these model-based baselines. Furthermore, for our analysis of nonlinear systems, we require access to damped system trajectories. It would be valuable to understand whether this is indeed necessary or whether our analysis could be extended to work without access to damped trajectories. 

As a final note, in this work we reduce the problem of stabilizing dynamical systems to running policy gradients on  a discounted LQR objective. This choice of reducing to LQR was in part made for simplicity to leverage previous analyses. However, it is possible that overall performance could be improved if rather than reducing to LQR, we instead attempted to run a model-free method that directly tries to optimize a robust control objective (which explicitly deals with uncertainty in the system dynamics).
 We believe that understanding these tradeoffs in objectives and their relevant sample complexities is an interesting avenue for future inquiry.

\section*{Acknowledgments}

We would like to thank Peter Bartlett for many helpful conversations and comments throughout the course of this project,
and Russ Tedrake for support with the numerical experiments.
JCP is supported by an NSF Graduate Research Fellowship. MS
is supported by an Open Philanthropy Fellowship grant.
JU is supported by the National Science Foundation, Award No. EFMA-1830901, 
and the Department of the Navy, Office of Naval Research, Award No. N00014-18-1-2210.

\newpage
\bibliographystyle{plainnat}
\bibliography{refs}
\newpage

\newpage

\appendix
\addtocontents{toc}{\protect\setcounter{tocdepth}{2}}
\renewcommand{\contentsname}{Table of Contents: Appendix}
\tableofcontents

\section{Deferred Proofs and Analysis for the Linear Setting}
\label{sec:linear_appendix}

\paragraph{Preliminaries on Linear Quadratic Control.} The cost a state-feedback controller $\Jlin(K \mid \gamma)$ is intimately related to the solution of the discrete-time Lyapunov equation. Given a stable matrix $\Acl$ and a symmetric positive definite matrix $\Sigma$, we define $\dlyap{\Acl}{\Sigma}$ to be the unique solution (over $X$) to the matrix equation, 
\begin{align*}
	X = \Sigma + \Acl^\top X \Acl.
\end{align*}
A classical result in Lyapunov theory states that $\dlyap{\Acl}{\Sigma}= \sum_{j=0}^\infty (\Acl^j)^\top X \Acl^j$. Recalling our earlier definition, for a controller $K$ such that $\sqrt{\gamma}(A+BK)$ is stable, we let $\Pky \defeq \dlyap{\sqrt{\gamma}(A+BK)}{Q + K^\top R K}$, where $Q,R$ are the cost matrices for the LQR problem defined in \Cref{def:lqr_objective}. As a special case, we let $\Pst \defeq \dlyap{A+B \Kst}{Q + \Kst^\top R \Kst}$ where $\Kst \defeq K_{\star, 1}$ is the optimal controller for the undiscounted LQR problem. Using these definitions, we have the following facts:
\begin{fact}
\label{fact:finite_cost_req}
$\Jlin(K \mid \gamma, A, B) < \infty$ if and only if $\sqrt{\gamma} \cdot (A + BK)$ is a stable matrix.
\end{fact}
\begin{fact}
\label{fact:value_functions}
If $\Jlin(K \mid \gamma, A, B) < \infty$ then for all $\matx \in \R^{\dimx}$, $\Jlin(K \mid  \matx, \gamma, A, B) = \matx^\top P_{K, \gamma} \matx$.
Furthermore, $$	\Jlin(K \mid \gamma, A, B)  = \Expop_{\matx_0 \sim \sqrt{\dimx} \dxsphere} \matx_0^\top \Pky \matx_0  =\traceb{P_{K, \gamma}}.
$$
\end{fact}

Employing these identities, we can now restate and prove \Cref{lemma:equivalence_lemma}:

\newtheorem*{lemma:equivalence_lemma_inf}{\Cref{lemma:equivalence_lemma}}
\begin{lemma:equivalence_lemma_inf} 
For all $K$ such that $J(K\mid \gamma, A,B) < \infty$, $\Jlin(K \mid \gamma, A, B) = \Jlin(K\mid 1, \sqrt{\gamma} \cdot A, \sqrt{\gamma} \cdot B)$.
\end{lemma:equivalence_lemma_inf}

\begin{proof}
From the definition of the LQR cost and linear dynamics in \Cref{def:lqr_objective}, for $\Acl \defeq A + BK$,
\begin{align*}
	\Jlin(K \mid \gamma, A, B) &= \E_{\matx_0}\left[ \sum_{t=0}^\infty \gamma^t \cdot  \matx_0^\top \left(\Acl^t\right)^\top  \left(Q  + K^\top R K\right) \Acl^t \matx_0 \right] \\
	&=  \E_{\matx_0}\left[ \sum_{t=0}^\infty   \matx_0^\top \left((\sqrt{\gamma}\Acl)^t\right)^\top  \left(Q  + K^\top R K\right) \left(\sqrt{\gamma}\Acl\right)^t \matx_0 \right] \\
	& = \Jlin(K \mid 1, \sqrt{\gamma} \cdot A, \sqrt{\gamma} \cdot B).
\end{align*}
\end{proof}

Therefore, since LQR with a discount factor is equivalent to LQR with damped dynamics, it follows that for $\gamma_0 < \rho(A)^{-2}$, (noisy) policy gradients initialized at the zero controller converges to the global optimum of the discounted LQR problem. The following lemma is essentially a restatement of the finite sample convergence result for gradient descent on LQR (Theorem 31) from \cite{fazel2018global}, where we have set  $Q=R=I$ and $\E \left[\matx_0 \matx_0^\top \right] = I$ as in the discount annealing algorithm. We include the proof of this result in \Cref{subsec:proof_noisy_pg} for the sake of completeness.

\begin{lemma}[\cite{fazel2018global}]
\label{lemma:noisy_pg_convergence}
For $K_0$ such that $\Jlin(K_0 \mid \gamma)< \infty$, define (noisy) policy gradients as the procedure which computes updates according to,
\begin{align*} 
K_{t+1} = K_t - \eta \widetilde{\grad}_{t}, 
\end{align*}
for some matrix $\noisygrad_t$. There exists a choice of a constant step size  $\eta > 0$ and a fixed polynomial,
\begin{align*}
\calC_{\mathrm{PG}} \defeq \poly\left( \frac{1}{\Jlin(K_0 \mid \gamma)}, \frac{1}{\opnorm{A}}, \frac{1}{\opnorm{B}}\right),
\end{align*}
such that if the following inequality holds for all $t = 1 \dots T$,  
\begin{align}
\label{eq:pg_grad_condition}
	\fronorm{\grad_{K} \Jlin(K_t \mid \gamma) - \noisygrad_t} \leq  \calC_{\mathrm{PG}} \cdot \epsilon,
\end{align}
 then $\Jlin(K_{T} \mid \gamma) - \min_{K} \Jlin(K \mid \gamma) \leq \epsilon$ whenever
\begin{align*}
	T \geq  \poly\Big(\opnorm{A}, \opnorm{B}, \Jlin(K_0 \mid \gamma)\Big) \log \left( \frac{\Jlin(K_0 \mid \gamma) - \min_{K} \Jlin(K \mid \gamma)}{\epsilon} \right).
\end{align*}
\end{lemma}

With this lemma, we can now present the proof of \Cref{theorem:linear_algorithm}. 

\subsection{Discount annealing on linear systems: Proof of \Cref{theorem:linear_algorithm}}

We organize the proof into two main parts. First, we prove statements $c)$ and $d)$ by an inductive argument. Then, having already proved part $b)$ in the main body of the paper, we finish the proof by establishing the stability guarantees of the resulting controller outlined in part $a)$.

\subsubsection{Proof of $c)$ and $d)$}
\paragraph{Base case.} Recall that at iteration $t$ of discount annealing, policy gradients is initialized at $K_{t,0} \defeq K_t$ and (in the case of linear systems) computes updates according to: 
\begin{align*}
	K_{t,j+1} = K_{t,j} - \eta \cdot \epsgrad\left( 
	\Jlin(\cdot \mid \gamma_t), K_{t,j}\right).
\end{align*}
From \Cref{lemma:noisy_pg_convergence}, if $\Jlin(K_{t,0} \mid \gamma_t) < \infty$ and 
\begin{align}
\label{eq:theorem1_gradclose}
	\fnorm{\epsgrad\left( 
	\Jlin(\cdot \mid \gamma_t), K_{t,j}\right) - \grad \Jlin(K_{t,j} \mid \gamma_t) } \leq \calC_{\mathrm{pg},t} \dimx \quad \forall j \geq 0,
\end{align}
where $\calC_{\mathrm{pg},t} = \poly(1/\Jlin(K_{t,0} \mid \gamma_t), 1 / \opnorm{A}, 1 / \opnorm{B}))$then policy gradients will converge to a $K_{t,j}$ such that $\Jlin(K_{t,j} \mid \gamma_t) - \min_K \Jlin(K \mid \gamma_t) \leq \dimx$ after $\poly(\opnorm{A}, \opnorm{B},\Jlin(K_{t,0}))$ many iterations.

By our choice of discount factor, we have that $\Jlin(0 \mid \gamma_0) < \infty$. Furthermore, since $\epsilon \leq \poly(1/\Jlin(0 \mid \gamma_0), 1 / \opnorm{A}, 1 / \opnorm{B}))$, then the condition outlined in \Cref{eq:theorem1_gradclose} holds and policy gradients achieves the desired guarantee when $t=0$.

The correctness of binary search at iteration $t=0$ follows from \Cref{lemma:binary_search}. In particular, we instantiate the lemma with $f = \Jlin(K_{t+1} \mid \cdot)$ of the LQR objective, which is nondecreasing in $\gamma$ by definition, $f_1 = 2.5 \Jlin(K_{t+1} \mid \gamma_t)$ and $f_2 = 8 \Jlin(K_{t+1} \mid \gamma_t)$. The algorithm requires auxiliary values $\fbar_1 \in [2.5 \Jlin(K_{t+1} \mid \gamma_t), 3 \Jlin(K_{t+1} \mid \gamma_t)]$ and $\fbar_2 \in [7 \Jlin(K_{t+1} \mid \gamma_t), 7.5 \Jlin(K_{t+1} \mid \gamma_t)]$ which we can always compute by using $\epseval$ to estimate the cost $\Jlin(K_{t+1} \mid \gamma_t)$ to precision $.1 \dimx$ (recall that $\Jlin(K \mid \gamma) \geq \dimx$ for any $K$ and $\gamma$). The last step needed to apply the lemma is to lower bound the width of the feasible region of $\gamma$'s which satisfy the desired criterion that $\Jlin(K_{t+1} \mid \gamma) \in [\fbar_1, \fbar_2]$. 

Let $\gamma' \geq \gamma_t$ be such that $\Jlin(K_{t+1} \mid \gamma') = 3 \Jlin(K_{t+1} \mid \gamma_t)$. Such a $\gamma$ is guaranteed to exist since $\Jlin(K_{t+1} \mid \cdot)$ is nondecreasing in $\gamma$ and it is a continuous function for all $\gamma < \bar{\gamma} \defeq \sup \{\gamma : \Jlin(K_{t+1} \mid \gamma) < \infty \}$. By the calculation from the proof of part $b)$ presented in the main body of the paper, for 
\begin{align*}
\gamma'' \defeq  \left( \frac{1}{8 \opnorm{P_{K_{t+1}, \gamma'}}^4} + 1 \right)^2 \gamma', 
\end{align*}
we have that $\Jlin(K_{t+1} \mid \gamma'') \leq 2 \Jlin(K_{t+1} \mid \gamma') = 6 \Jlin(K_{t+1} \mid \gamma_t)$. By monotonicity and continuity of $\Jlin(K_{t+1} \mid \cdot)$, when restricted to $\gamma \leq \bar{\gamma}$, all $\gamma \in [\gamma', \gamma'']$ satisfy $\Jlin(K_{t+1} \mid \gamma) \in [3\Jlin(K_{t+1} \mid \gamma_t), 6\Jlin(K_{t+1} \mid \gamma_t)]$. Moreover, 
\begin{align*}
	\gamma'' - \gamma' = \left[ \left( \frac{1}{8 \opnorm{P_{K_{t+1}, \gamma'}}^4} + 1 \right)^2 - 1 \right] \gamma' 
	\geq \frac{1}{4 \opnorm{P_{K_{t+1}, \gamma'}}^4} \gamma' \geq \frac{1}{4 \traceb{P_{K_{t+1}, \gamma'}}^4} \gamma_t,
\end{align*}
where the last line follows from the fact that $\gamma' \geq \gamma_t$ and that the trace of a PSD matrix is always at least as large as the operator norm. Lastly, since $\traceb{P_{K_{t+1}, \gamma'}} = \Jlin(K_{t+1} \mid \gamma') = 3 \Jlin(K_{t+1} \mid \gamma_t)$, by the guarantee of policy gradients, we have that for $t=0$, $\Jlin(K_{1} \mid \gamma_0) \leq 2 \traceb{\Pst}$. Therefore, for $t=0$:
\begin{align*}
	\gamma'' - \gamma' \geq \frac{1}{4 (6 \traceb{\Pst})^4}
\end{align*}
Hence the width of the feasible region is at least $\frac{1}{5184 \traceb{\Pst}^4} \gamma_0$. 

\paragraph{Inductive step.}

To show that policy gradients achieves the desired guarantee at iteration $t+1$, we can repeat the exact same argument as in the base case. The only difference is that the need to argue that the cost of the initial controller $\sup_{t\geq 0}\Jlin(K_t \mid \gamma_t)$ is uniformly bounded across iterations. By the inductive hypothesis on the success of the binary search algorithm at iteration $t-1$, we have that,
\begin{align*}
	\Jlin(K_{t} \mid \gamma_{t}) & \leq 8 \Jlin(K_{t} \mid \gamma_{t-1}) \\
	& \leq 8\left( \min_{K} \Jlin(K \mid \gamma_{t-1}) + \dimx\right) \\
	& \leq 8\left( \min_{K} \Jlin(K \mid 1) + \dimx\right) \\
	& \leq 16 \traceb{P_\star}.
\end{align*} 
Hence, by \Cref{lemma:noisy_pg_convergence}, policy gradients achieves the desired guarantee using $\poly(M_{\mathrm{lin}}, \opnorm{A} , \opnorm{B})$ many queries to $\epsgrad(\cdot , \Jlin(\cdot \mid \gamma))$ as long as $\epsilon$ is less than $ \poly(M_{\mathrm{lin}}^{-1}, \opnorm{A}^{-1}, \opnorm{B}^{-1})$.

Likewise, the argument for the correctness of the binary search procedure is identical to that of the base case. Because of the success of policy gradients and binary search at the previous iteration, we can upper bound, $\traceb{P_{K_{t+1}, \gamma'}}$ by $6 \traceb{\Pst}$ and get a uniform lower bound on the width of the feasible region.

\subsubsection{Proof of $a)$}
After halting, we see that discount annealing returns a controller $\Khat$ satisfying the stated condition from Step 2 requiring that, 
\begin{align*}
	\Jlin(\Khat \mid 1) - \Jlin(\Kst \mid 1) = \traceb{\Phat - \Pst}\leq \dimx.
\end{align*}
Here, we have used \Cref{fact:value_functions} to rewrite $\Jlin(\Khat \mid 1)$ as $\trace [\Phat]$ for $\Phat \defeq \dlyap{A+B\Khat}{Q +\Khat^\top R \Khat}$ (and likewise for $\Pst$). Since $\traceb{\Pst} \geq \dimx$, we conclude that $\trace[\Phat] \leq 2 \traceb{\Pst}$. Now, by properties of the Lyapunov equation (see \Cref{lemma:dlyap}) the following holds for $\Acl \defeq A+B\Khat$: 
\begin{align*}
	\opnorm{\Acl^2}^t \leq \opnorm{\Phat} \left(1 - \frac{1}{\opnorm{\Phat}}\right)^t \leq \trace[\Phat] \exp\left( - t  / \trace [\Phat] \right).
\end{align*}
Hence, we conclude that, 
\begin{align*}
	\norm{\matx_{t}} = \norm{\Acl^t \matx_0} \leq \opnorm{\Acl^t} \norm{\matx_0} \leq \sqrt{2\trace[\Pst]} \exp\left( - \frac{1}{4\traceb{\Pst}} \cdot  t\right) \norm{\matx_0}.
\end{align*}

\subsection{Convergence of policy gradients for LQR: Proof of \Cref{lemma:noisy_pg_convergence}}
\label{subsec:proof_noisy_pg}
\begin{proof}
Note that, by \Cref{lemma:equivalence_lemma}, proving the above result for $\Jlin(\cdot \mid \gamma, A , B)$ is the same as proving it for $\Jlin(\cdot \mid 1, \sqrt{\gamma}A, \sqrt{\gamma}B)$. We start by defining the following idealized updates,
\begin{align*}
	K'  &= K - \eta \grad_K \Jlin(K_t \mid \gamma) \\ 
	K''  &= K - \eta \noisygrad_t.
\end{align*}
From Lemmas 13, 24, and 25 from \cite{fazel2018global}, there exists a fixed polynomial, 
\begin{align*}
\calC_{\eta} \defeq \poly\left( \frac{1}{\sqrt{\gamma} \opnorm{A}}, \frac{1}{\sqrt{\gamma}\opnorm{ B}}, \frac{1}{\Jlin(K_0 \mid \gamma)}  \right)
\end{align*}
Such that, for $\eta \leq \calC_{\eta}$, the following inequality holds,
\begin{align*}
	\Jlin(K' \mid \gamma) - \Jlin(K_\star \mid \gamma) \leq \left(1 - \eta \frac{1}{\opnorm{\Sigma_{K_\star}}} \right) \left( \Jlin(K \mid \gamma) - \Jlin(K_\star \mid \gamma) \right),
\end{align*}
where $\Sigma_{\Kst} = \E_{\matx_0}[\sum_{t=0}^\infty \matx_{t, \star}^\top \matx_{t,\star}]$ and $\{ \matx_{t,\star}\}_{t=0}^\infty$ is the sequence of states generated by the controller $\Kst$. Therefore, if $\Jlin(K'' \mid \gamma)$ and $\Jlin(K' \mid \gamma)$ satisfy,
\begin{align}
\label{eq:closeness_condition}
	| \Jlin(K'' \mid \gamma) - \Jlin(K' \mid \gamma) | \leq \frac{1}{2\opnorm{\Sigma_{\Kst}}} \eta \epsilon
\end{align}
then, as long as $\Jlin(K \mid \gamma) - \Jlin(K_\star \mid \gamma) \geq \epsilon$, this following inequality also holds:
\begin{align*}
	\Jlin(K'' \mid \gamma) - \Jlin(K_\star \mid \gamma) \leq \left(1 - \eta \frac{1}{\opnorm{2\Sigma_{K_\star}}} \right) \left( \Jlin(K \mid \gamma) - \Jlin(K_\star \mid \gamma) \right).
\end{align*}
The proof then follows by unrolling the recursion and simplifying. We now focus on establishing \Cref{eq:closeness_condition}. By Lemma 27 in \cite{fazel2018global}, if 
\begin{align*}
	\opnorm{K'' - K'} = \eta \opnorm{\grad_{K} \Jlin(K_j \mid \gamma) - \noisygrad_j} \leq \calC_K  
\end{align*}
where $\calC_K$ is a fixed polynomial $\calC_K \defeq \poly(\frac{1}{\Jlin(K_0 \mid \gamma)},\frac{1}{\sqrt{\gamma} \opnorm{A}}, \frac{1}{\sqrt{\gamma}\opnorm{B}})$, then 
\begin{align*}
|\Jlin(K'' \mid \gamma) - \Jlin(K' \mid \gamma) | \leq \calC_{\mathrm{cost}} \opnorm{K'' - K'} = \calC_{\mathrm{cost}} \cdot \eta \opnorm{\grad_{K} \Jlin(K_j \mid \gamma) - \noisygrad_j}, 
\end{align*}
where $\Ccost \defeq \poly(\dimx, \opnorm{R}, \opnorm{B},  \Jlin(K_0))$. Therefore,  \Cref{eq:closeness_condition} holds if 
\begin{align*}
 \opnorm{\grad_{K} \Jlin(K_j \mid \gamma) - \noisygrad_j} \leq  \min\left\{\frac{1}{2\opnorm{\Sigma_{\Kst}} \calC_{\mathrm{cost}}} \epsilon, \; \calC_K\right\}.
\end{align*}
The exact statement follows from using $\dimx \leq \Jlin(K_0 \mid \gamma)$ and $\opnorm{\Sigma_{K_{\star, \gamma}}} \leq \Jlin(K_{\star, \gamma}) \leq \Jlin(K_0 \mid \gamma)$ by Lemma 13 in \citep{fazel2018global} and taking the polynomial in the proposition statement to be the minimum of $\calC_K$ and $1 / \Ccost $.
\end{proof}

\subsection{Impossibility of reward shaping: Proof of  \Cref{prop:reward_shaping}}
\label{subsec:proof_of_reward_shaping}
\begin{proof}
Consider the linear dynamical system with dynamics matrices, 
\begin{align*}
	A = \begin{bmatrix}
		0 & 0 \\ 
		0& 2 
	\end{bmatrix}, \quad B = \begin{bmatrix}
		1 \\ 
		\beta
	\end{bmatrix}
\end{align*}
where $\beta > 0 $ is a parameter to be chosen later. Note that a linear dynamical system of these dimensions is controllable (and hence stabilizable) \citep{callier2012linear}, since the matrix 
\begin{align*}
	\begin{bmatrix}
		B & AB 
	\end{bmatrix} = \begin{bmatrix}
		1 & 0 \\ 
		\beta & 2\beta
	\end{bmatrix}
\end{align*}
is full rank. For any controller $K = \begin{bmatrix}
	k_1 & k_2 \end{bmatrix}$, the closed loop system $\Acl \defeq A + BK$ has the form, 
	\begin{align*}
		\Acl = \begin{bmatrix}
			k_1 & k_2 \\ 
			\beta k_1 & 2 + \beta k_2 
		\end{bmatrix}.
	\end{align*}
By Gershgorin's circle theorem, $
\Acl$ has an eigenvalue $\lambda$ which satisfies, 
\begin{align*}
	|\lambda| \geq |2 + \beta k_2| - |\beta k_1| \geq 2 - 2\beta \max\{ |k_1|, |k_2|\}.
\end{align*}
Therefore, any controller $K$ for which the closed-loop system $A+BK$ is stable must have the property that, 
\begin{align*}
	\max\{|k_1|, |k_2|\} \geq \frac{1}{2 \beta}. 
\end{align*}
Using this observation and \Cref{fact:value_functions}, for any discount factor $\gamma$,  a stabilizing controller $K$ must satisfy, 
\begin{align*}
	\Jlin(K \mid \gamma) &= \traceb{\Pky} \\
	& \geq \traceb{Q} + \traceb{K^\top R K} \\ 
	& \geq (k_1^2 + k_2^2) \cdot  \sigma_{\min}(R) \\ 
	& \geq \frac{1}{4 \beta^2} \cdot \sigma_{\min}(R).
\end{align*}
In the above calculation, we have used the identity $\Pky = Q + K^\top R K + \gamma\Acl^\top \Pky \Acl$ as well as the assumption that $R$ is positive definite. Next, we observe that for a discount factor $\gamma = c^2 \cdot \rho(A)^{-2}$, where $c \in (0,1)$ as chosen in the initial iteration of our algorithm, the cost of the 0 controller has the following upper bound: 
\begin{align*}
	\Jlin(0 \mid \gamma_0) &= \sum_{j=0}^\infty c^{2j} \rho(A)^{-2j}\cdot \traceb{(A^\top)^j Q A^j} \\ 
	& \leq \opnorm{Q} \sum_{j=0}^\infty  c^{2j} \rho(A)^{-2j}  \fnorm{A^j}^2 \\ 
	& = \opnorm{Q} \sum_{j=0}^\infty \fronorm{\left( \frac{c}{\rho(A)} \cdot A \right)^j}^2.
\end{align*}
Using standard Lyapunov arguments (see for example Section D.2 in \cite{perdomo2021dimensionfree}) the sum in the last line is a geometric series and is equal to some function $f(c, A) < \infty$, which depends \emph{only} on $c$ and $A$, for all $c \in (0,1)$. Using this calculation, it follows that 
\begin{align*}
	\min_{K} \Jlin(K \mid \gamma_0) \leq \Jlin(0 \mid \gamma) \leq \opnorm{Q} f(c,A)
\end{align*}
Hence, for any $Q, R$, and discount factor $\gamma \in (0, \rho(A)^{-2})$, we can choose $\beta$ small enough such that, 
\begin{align*}
\opnorm{Q} f(c,A) < \frac{1}{4\beta^2}  \sigma_{\min}(R)
\end{align*}
implying that the optimal controller $K_{\star, \gamma}$ for the discounted problem cannot be stabilizing for $(A,B)$.
\end{proof}

\subsection{Auxiliary results for linear systems}
\label{subsec:supporting_lemmas}
\begin{proposition}
\label{lemma:lower_bound_gamma}
Let $\sqrt{\gamma_1}(A+BK)$ be a stable matrix and define $P_1 \defeq \dlyap{\sqrt{\gamma_1}(A+BK)}{Q+K^\top RK}$, then for $c$ defined as
\begin{align*}
	c \defeq  \frac{1}{8\opnorm{P_1}^{4}} + 1,
\end{align*}
the following holds for $\gamma_2 \defeq c^2 \gamma_1$ and $P_2 \defeq \dlyap{\sqrt{\gamma_2}(A+BK)}{Q+K^\top RK}$:
\begin{align*}
\traceb{P_2 - P_1} \leq \traceb{P_1}.
\end{align*}
\end{proposition}
\begin{proof}
The proof is a direct consequence of Proposition C.7 in \cite{perdomo2021dimensionfree}. In particular, we use their results for the trace norm and use the following substitutions,
\begin{align*}
&A_1 \leftarrow \sqrt{\gamma_1}(A+BK)\quad &\Sigma \leftarrow Q + K^\top R K  \\ 
&A_2 \leftarrow \sqrt{\gamma_2}(A+BK) \quad&\alpha \leftarrow 1/2
\end{align*}
where $A_1, A_2, \Sigma,$ and $\alpha$ are defined as in \cite{perdomo2021dimensionfree}. Note that for $c$ satisfying,
\begin{align*}
c \leq  \frac{1}{8 \opnorm{\sqrt{\gamma_1}(A+BK)}} \min \left\{\frac{1}{ \opnorm{P_1}^{3/2}};\; \frac{\traceb{P_1}}{\dimx \opnorm{P_1}^{7/2} }\right\} + 1
\end{align*}
we get that,
\begin{align*}
	\opnorm{A_1 - A_2}^2 = \gamma_1 (c-1)^2 \opnorm{A+BK}^2 \leq   \frac{1}{64 \opnorm{P_1}^3} = \frac{\alpha^2}{16 \opnorm{P_1}^3}.
\end{align*}
Therefore, Proposition C.7 states that, for $\calC \defeq \traceb{P_1^{-1/2}(Q+K^\top R K) P_1^{-1/2}} \leq \traceb{I} = \dimx$,
\begin{align*}
\traceb{P_2 - P_1} &\leq 8 \calC \sqrt{\gamma_1} (c-1) \opnorm{A+BK} \opnorm{P_1}^{7/2} \\
& \leq \traceb{P_1}.
\end{align*}
Lastly, noting that,
\begin{align*}
P_1 = \dlyap{\sqrt{\gamma_1}(A+BK)}{Q+K^\top RK} \succeq \gamma  (A+BK)^\top  (A+BK)
\end{align*}
we have that $\opnorm{\sqrt{\gamma_1}(A+BK)} \leq \opnorm{P_1}^{1/2}$ and $\traceb{P_1} \geq \dimx $. Therefore, since $\opnorm{P_1} \geq 1$, in order to apply Proposition C.7 from \cite{perdomo2021dimensionfree} it suffices for $c$ to satisfy,
\begin{align*}
	c \leq \frac{1}{8 \opnorm{P_1}^4} + 1.
\end{align*}
\end{proof}
\begin{lemma}[Lemma D.9 in \cite{perdomo2021dimensionfree}]
\label{lemma:dlyap}
Let $A$ be a stable matrix, $Q \succeq I$, and define $P \defeq \dlyap{A}{Q}$. Then, for all $j \geq 0$, 
\begin{align*}
(A^\top)^j  A^j  \preceq (A^\top)^j P A^j \preceq P \left(1 - \frac{1}{\opnorm{P}}\right)^j.
\end{align*}
\end{lemma}

\begin{lemma}
\label{lemma:stability_margin}
Let $A$ be a stable matrix and define $P \defeq \dlyap{A}{Q}$ where $Q \succeq I$. Then, for any matrix $\Delta$ such that $\opnorm{\Delta} \leq \frac{1}{6 \opnorm{P}^2}, 
$
it holds that for all $j \geq 0$
\begin{align*}
	\left( (A + \Delta)^\top \right)^j P  (A+\Delta)^j \preceq P \left(1 - \frac{1}{2\opnorm{P}}\right)^j.
\end{align*}
\end{lemma}
\begin{proof}
Expanding out, we have that 
\begin{align*}
	(A + \Delta)^\top P (A + \Delta) &= A^\top P A + A^\top P \Delta + \Delta^\top P A + \Delta^\top P \Delta \\ 
	& \preceq P \left(1 - \frac{1}{\opnorm{P}}\right) + \opnorm{A^\top P \Delta + \Delta^\top P A + \Delta^\top P \Delta} I,
\end{align*}
where in the second line we have used properties of the Lyapunov function, \Cref{lemma:dlyap}. Next,  we observe that 
\begin{align*}
	\opnorm{\Delta^\top P A} \leq \opnorm{\Delta^\top P^{1/2}} \opnorm{P^{1/2} A} \leq \opnorm{\Delta^\top P^{1/2}}  \opnorm{P^{1/2}} \leq \opnorm{\Delta} \opnorm{P},
\end{align*}   
where we have again used \Cref{lemma:dlyap} to conclude that $A^\top P A \preceq P$. Note that the exact same calculation holds for $\opnorm{A^\top P \Delta}$. Hence, we can conclude that for $\Delta$ such that $\opnorm{\Delta} \leq 1$,
\begin{align*}
	(A + \Delta)^\top P (A + \Delta) \preceq P \left(1 - \frac{1}{\opnorm{P}}\right) + 3 \opnorm{P} \opnorm{\Delta}  \cdot I.
\end{align*}
Using the fact that, $P \succeq I $ and that $\opnorm{\Delta} \leq 1 / (6 \opnorm{P}^2)$, we get that, 
\begin{align*}
3 \opnorm{P} \opnorm{\Delta}  \cdot I \preceq \frac{1}{2\opnorm{P}} P,
\end{align*}
which finishes the proof.
\end{proof}


\section{Deferred Proofs and Analysis for the Nonlinear Setting}

\textbf{Establishing Lyapunov functions.} Our analysis for nonlinear systems begins with the observation that any state-feedback controller $K$ which achieves finite cost on the $\gamma$-discounted LQR problem has an associated value function $\Pky$ which can be used as a Lyapunov function for the $\sqrt{\gamma}$-damped nonlinear dynamics, for small enough initial states. We present the proof of this result in \Cref{subsec:proof_of_nonlinear_lyapunov}.

\begin{lemma}
\label{lemma:nonlinear_lyapunov}
Let $\Jlin(K \mid \gamma)<\infty$. Then, for all $\matx \in \R^\dimx$ such that,
\begin{align*}
	\matx^\top \Pky \matx \leq \frac{\rnl^2}{4 \czero^2 \opnorm{\Pky}^3},
\end{align*}
the following inequality holds: 
\begin{align*}
	\gamma \cdot \Gnl(\matx, K\matx)^\top P_{K, \gamma} \Gnl(\matx, K\matx) \leq \matx^\top P_{K,\gamma} \matx \cdot \left( 1 - \frac{1}{2\opnorm{P_{K,\gamma}}} \right).
\end{align*}
\end{lemma}

Using this observation, we can then show that any controller which has finite discounted LQR cost is exponentially stabilizing over states in a sufficiently small region of attraction.
\begin{lemma}
\label{lemma:exp_convergence}
Assume $\Jlin(K \mid \gamma) < \infty$ and define $\{\xtnl\}_{t=0}^\infty$ be the sequence of states generated according to $\xtnlp = \sqrt{\gamma}\Gnl(\xtnl, K \xtnl)$ where $\matx_{t, \mathrm{nl}} = \matx_0$. If $\matx_0$ is such that,
\begin{align*}
	V_0 := \matx_0^\top \Pky \matx_0 \leq \frac{\rnl^2}{4 \czero^2 \opnorm{\Pky}^3} ,
\end{align*}
then for all $t\geq 0$ and for $V_0$ defined as above,
\begin{enumerate}[a)]
\item The norm of the state $\norm{\xtnl}^2$ is bounded by
\begin{align}
	\norm{\xtnl}^2 \leq \xtnl^\top P_{K,\gamma} \xtnl \leq V_0 \left(1 - \frac{1}{2 \opnorm{\Pky}} \right)^t. \label{eq:corollary:exp_convergence}
\end{align}
\item The norms of $\fnl(\xtnl,K\xtnl)$ and $\nabla \fnl(\xtnl,K\xtnl)$ are bounded by  
\begin{align}
\|\fnl(\xtnl,K\xtnl)\| &\le \betanl (1+\opnorm{K}^2)  V_0 \left(1 - \frac{1}{2 \opnorm{\Pky}} \right)^t. \label{eq:fnl_bound}\\
\opnorm{\nabla \fnl(\xtnl,K\xtnl)} &\le \beta (1+\|K\|_{\op})  V_0^{1/2} \left(1 - \frac{1}{2 \opnorm{\Pky}} \right)^{t/2}. \label{eq:grad_fnl_bound}
\end{align}
\end{enumerate}

\end{lemma}
\begin{proof}
 The proof of $(a)$ follows by repeatedly applying \Cref{lemma:nonlinear_lyapunov}. Part $(b)$ follows from the first after using \Cref{lemma:jacobian}. The statement of the lemma in the main body follows from using $\opnorm{\Pky} \leq \traceb{\Pky} = \Jlin(K \mid \gamma)$ and simplifying $(1-x)^t \leq \exp(-tx)$.
\end{proof}

\textbf{Relating $\Gnl$ to its Jacobian Linearization.} Having established how any controller that achieves finite LQR cost is guaranteed to be stabilizing for the nonlinear system, we now go a step further and illustrate how this stability guarantee can be used to prove that the difference in costs and gradients between $\Gnl$ and its Jacobian linearization are guaranteed to be small.

\newtheorem*{prop:diff_lin_nl_restated}{\Cref{prop:diff_lin_nl} (restated)}
\begin{prop:diff_lin_nl_restated} 
Assume $\Jlin(K \mid \gamma) < \infty$. Then, 
\begin{enumerate}[a)]
	\item If $r \leq \frac{\rnl}{2 \czero \opnorm{\Pky}^2}$, then 	$\big| \Jnl(K \mid \gamma, r) - \Jlin(K \mid \gamma)\big| \leq 8  \dimx \czero \opnorm{\Pky}^4 \cdot r.$

	\item If $r \leq \frac{\rnl}{12 \betanl \opnorm{\Pky}^{5/2}}$, then, 
	\begin{align*}
		\fnorm{\grad_K \Jnl(K \mid \gamma, r) - \grad_K \Jlin(K \mid \gamma)} \leq 48 \dimx  \czero (1 + \opnorm{B})  \opnorm{\Pky}^{7} r.
	\end{align*}
\end{enumerate}
\end{prop:diff_lin_nl_restated}

%
\begin{proof}
Due to our assumption on $r = \norm{\matx_0}$, we have that, 
\begin{align*}
	\matx_0^\top \Pky \matx_0 \leq \opnorm{\Pky} \norm{\matx_0}^2 \leq \frac{1}{4 \czero^2 \opnorm{\Pky}^3}.
\end{align*}
Therefore, we can apply \Cref{lemma:similarity_of_costs} to conclude that, 
\begin{align*}
\big|\Jnl(K \mid \gamma, \matx_0) - \Jlin(K \mid \gamma, \matx_0)\big| \leq 	8 \czero \norm{\matx_0}^3 \opnorm{\Pky}^4.
\end{align*}
Next, we multiply both sides by $\dimx / r^2$, take expectations, and apply Jensen's inequality to get that, 
\begin{align*}
	\big| \frac{\dimx }{r^2} \E_{\matx_0 \sim r \cdot \dxsphere}\Jnl(K \mid \gamma, \matx_0)  - \frac{\dimx}{r^2} \E_{\matx_0 \sim r \cdot \dxsphere} \Jlin(K \mid \gamma, \matx_0) \big| \leq 	8 \frac{\dimx}{r^2} \czero \opnorm{\Pky}^4 \E_{\matx_0 \sim r \cdot \dxsphere}\norm{\matx_0}^3.
\end{align*}
Given our definitions of the linear objective in 
\Cref{def:lqr_objective}, we have that,
\[
 \frac{\dimx}{r^2} \E_{\matx_0 \sim r \cdot \dxsphere} \Jlin(K \mid \gamma, \matx_0) = \Jlin(K \mid \gamma),
\] 
for all $r > 0$. Therefore, we can rewrite the inequality above as,
\begin{align*}
	\big| \Jnl(K \mid \gamma, r) - \Jlin(K \mid \gamma)\big| \leq 8  \dimx \czero \opnorm{\Pky}^4 \cdot r.
\end{align*}
The second part of the proposition uses the same argument as part a, but this time employing \Cref{lemma:grad_pointwise_closeness} to bound the difference in gradients (pointwise).
\end{proof}

In short, this previous lemma states that if the cost on the linear system is bounded, then the costs and gradients between the nonlinear objective and its Jacobian linearization are close. We can also prove the analogous statement which establishes closeness while assuming that the cost on the nonlinear system is bounded. 

\begin{lemma}
\label{lemma:jnl_close}
Let $\alpha > 1$ be such that $80\dimx^2 \Jnl(K \mid \gamma, r) \leq \alpha$. 
\begin{enumerate}
	\item If $r \leq \frac{\rnl^2}{64\betanl \alpha^2 (1 + \opnorm{K})}$, then $|\Jnl(K \mid \gamma, r) - \Jlin(K\mid \gamma)| \leq 8 \dimx \betanl \alpha^4 r$.
	\item 	If $r \leq \frac{1}{12 \betanl \alpha^{5/2}}$, then $	\fnorm{\grad_K \Jnl(K \mid \gamma, r) - \grad_K \Jlin(K \mid \gamma)} \leq 48 \dimx  \czero (1 + \opnorm{B})  \alpha^{7} \cdot r.
$
\end{enumerate}
\end{lemma}

\begin{proof}
The lemma is a consequence of combining \Cref{prop:diff_lin_nl} and \Cref{prop:finite_horizon}.  In particular, from \Cref{prop:finite_horizon} if $r \leq \min \{ \alpha \rnl^2, \; \frac{\dimx}{64 \alpha^2 \betanl (1 + \opnorm{K})}\}$, then 
\begin{align*}
	80 \dimx^2 \Jnl(K \mid \gamma, r) \geq \min\{ \Jlin(K \mid \gamma), \alpha \}
\end{align*}
However, since $\alpha \geq 80 \dimx^2 \Jnl(K \mid \gamma, r)$, we conclude that $80 \dimx^2 \Jnl(K \mid \gamma, r) \geq  \Jlin(K \mid \gamma)$. Having shown that the linear cost is bounded, we can now plug in \Cref{prop:diff_lin_nl}. In particular, if 
\begin{align*}
	r \leq \frac{\rnl}{2\betanl \alpha^2} \leq \frac{\rnl}{2\betanl \Jlin(K \mid \gamma)^2}
\end{align*} 
then, \Cref{prop:diff_lin_nl} states that 
\begin{align*}
	|\Jnl(K \mid \gamma, r) - \Jlin(K \mid \gamma)| \leq 8 \dimx \betanl \Jlin(K \mid \gamma)^4 r \leq 8 \dimx \betanl \alpha^4 r. 
\end{align*}
To prove the second part of the statement, we again use \Cref{prop:diff_lin_nl}. In particular, since
\begin{align*}
	r \leq \frac{1}{12 \betanl \alpha^{5/2}} \leq \frac{1}{12 \betanl \Jlin(K \mid \gamma)^{5/2}}
\end{align*}
we can hence conclude that 
\begin{align*}
	\fnorm{\grad_K \Jnl(K \mid \gamma, r) - \grad_K \Jlin(K \mid \gamma)} \leq 48 \dimx  \czero (1 + \opnorm{B})  \alpha^{7} \cdot r.
\end{align*}

\end{proof}

\subsection{Discount annealing on nonlinear systems: Proof of \Cref{theorem:nonlinear_algorithm}}

As in \Cref{theorem:linear_algorithm}, we first prove parts $c)$ and $d)$ by induction and then prove parts $a)$ and $b)$ separately. 

\subsubsection{Proof of $c)$ and $d)$}

\textbf{Base case.} As before, at each iteration $t$ of discount annealing, policy gradients is initialized at $K_{t,0} \defeq K_t $ and computes updates according to, 
\begin{align*}
	K_{t,j+1} = K_{t,j} - \eta \cdot \epsgrad\left( 
	\Jnl(\cdot \mid \gamma_t, \rst), K_{t,j}\right).
\end{align*}
To prove correctness, we show that the noisy gradients on the nonlinear system are close to the true gradients on the \emph{linear} system. That is, 
\begin{align}
\label{eq:theorem2_gradclose}
	\fnorm{\epsgrad\left( 
	\Jnl(\cdot \mid \gamma_t, \rst), K_{t,j}\right) - \grad \Jlin(K_{t,j} \mid \gamma_t) } \leq \calC_{\mathrm{pg},t} \dimx \quad \forall j \geq 0,
\end{align}
where $\calC_{\mathrm{pg},t} = \poly(1/\opnorm{A}, 1/\opnorm{B}, 1/\Jlin(K_{t} \mid \gamma_t))$ is again a fixed polynomial from \Cref{lemma:noisy_pg_convergence}. 

Consider the first iteration of discount annealing, by choice of $\gamma_0$, we have that $\Jlin(K_0 \mid \gamma_0) < \infty$. Therefore, by \Cref{prop:diff_lin_nl} if 
\begin{align*}
	\rst \leq  \min\left\{ \frac{\rnl}{12 \betanl \Jlin(K_{0,0} \mid \gamma_0)}, \frac{  \calC_{\mathrm{pg},0}}{100   \czero (1 + \opnorm{B})  \Jlin(K_{0,0} \mid \gamma_0)^{7}}\right\}
\end{align*}
it must hold that $\fnorm{\grad
	\Jnl(K_{0,0} \mid \gamma_0, \rst) - \grad \Jlin(K_{0,0} \mid \gamma_0) } \leq .5\calC_{\mathrm{pg},0} \dimx$. Likewise, if we choose the tolerance parameter $\epsilon \leq .5\calC_{\mathrm{pg},0} \dimx$ in $\epsgrad$ then we have that $$\fnorm{\epsgrad
	\left( \Jnl(\cdot  \mid \gamma_0, \rst), K_{0,0} \right) - \Jnl(K_{0,0} \mid \gamma_0, \rst)} \leq .5\calC_{\mathrm{pg},0} \dimx.$$
By the triangle inequality, the inequality in \Cref{eq:theorem2_gradclose} holds for $t=0$ and $j=0$. However, because \Cref{lemma:noisy_pg_convergence} shows that policy gradients is a descent method, that is $\Jlin(K_{0,j} \mid \gamma_0) \leq \Jlin(K_{0,0} \mid \gamma_0)$ for all $j \geq 0$, \Cref{eq:theorem2_gradclose} also holds for all $j \geq 0 $ for the same choice of $\rst$ and tolerance parameter for $\epsgrad$. By guarantee of \Cref{lemma:noisy_pg_convergence}, for $t=0$, policy gradients achieves the guarantee outlined in Step 2 using at most $\poly(\opnorm{A}, \opnorm{B}, \Jlin(K_0 \mid \gamma_0))$ many queries.

To prove that random search achieves the guarantee outlines in Step 4 at iteration 0 of discount annealing, we appeal to \Cref{lemma:random_search}. In particular, we instantiate the lemma with $f \leftarrow \Jnl(K_1 \mid \cdot, \rst)$, $f_1 \leftarrow 8\Jnl(K_1 \mid \gamma_0, \rst)$, $f_2 \leftarrow 2.5 \Jnl(K_1 \mid \gamma_0, \rst)$. As before, the algorithm requires values $\fbar_1 \in [2.9 \Jnl(K_1 \mid \gamma_0, \rst), 3\Jnl(K_1 \mid \gamma_0)]$ and $\fbar_2 \in [6 \Jnl(K_1 \mid \gamma_0, \rst), 6.1\Jnl(K_1 \mid \gamma_0)]$. These can be estimated via two calls to $\epseval$ with tolerance parameter $.01 \dimx$.

To show the lemma applies we only need to lower bound the width of feasible $\gamma$ such that 
\begin{align}
\label{eq:rs_nl}
	2.9 \Jnl(K_1 \mid \gamma_0, \rst) \leq \Jnl(K_1 \mid \gamma, \rst) \leq 6.1  \Jnl(K_1 \mid \gamma_0, \rst)
\end{align} 
From the guarantee from policy gradients, we know that $\Jlin(K_1 \mid \gamma_0) \leq 2\traceb{\Pst}$. Furthermore, from the proof of \Cref{theorem:linear_algorithm}, we know that there exists $\gamma'', \gamma' \in [0,1]$ satisfying, $\gamma'' - \gamma' \geq \frac{1}{5200\traceb{\Pst}^4} \gamma_0$, such that for all $\gamma \in [\gamma', \gamma'']$ 
\begin{align}
\label{eq:rs_lin}
	3 \Jlin(K_1\mid \gamma_0) \leq \Jlin(K_1 \mid \gamma) \leq 6 \Jlin(K_1 \mid \gamma_0).
\end{align}
To finish the proof of correctness, we show that any $\gamma$ that satisfies \Cref{eq:rs_lin} must also satisfy \Cref{eq:rs_nl}. In particular, since $\Jlin(K_1 \mid \gamma_0) \leq 2\traceb {\Pst}$ and $\Jlin(K_1 \mid \gamma) \leq 12 \traceb{\Pst}$ for 
\begin{align*}
	\rst \leq \min \left\{ \frac{\rnl}{2\betanl (12 \traceb{\Pst})^2},\; \frac{.01}{8 \betanl (12 \traceb{\Pst})^4} \right\},
\end{align*}
it holds that $|\Jnl(K_1 \mid \gamma_0, \rst) - \Jlin(K_1 \mid \gamma_0)| \leq .01 \dimx$ and $|\Jnl(K_1 \mid \gamma, \rst) - \Jlin(K_1 \mid \gamma)| \leq .01 \dimx$. Using these two inequalities along with \Cref{eq:rs_lin} implies that \Cref{eq:rs_nl} must also hold. Therefore, the width of the feasible region is at least $1 / (5200 \traceb{\Pst}^4)$ and random search must return a discount factor using at most 1 over this many iterations by \Cref{lemma:random_search}.

\paragraph{Inductive step.} To show that policy gradients converges, we can use the exact same argument as when arguing the base case where instead of referring to $\Jlin(K_0 \mid \gamma_0)$ and $\Jnl(K_{0} \mid \gamma_0, \rst)$ we use $\Jlin(K_t \mid \gamma_t)$ and $\Jnl(K_{t} \mid \gamma_t, \rst)$. In particular, we can reuse the same argument as in the base case, but need to ensure that is that $\sup_{t \geq 1} \Jlin(K_t \mid \gamma_t)$ is uniformly bounded. 

To prove this, from the inductive hypothesis on the correctness of binary search at previous iterations, we know that $\Jnl(K_t \mid \gamma_t) \leq 8 \Jnl(K_t \mid \gamma_{t-1}, \rst)$. Again by the inductive hypothesis, at time step $t-1$ policy gradients achieves the desired guarantee from Step 2, implying that $\Jlin(K_t \mid \gamma_{t-1}) \leq 2 \traceb{\Pst}$. By choice of $\rst$ this implies that $$|\Jlin(K_t \mid \gamma_{t-1}) - \Jnl(K_t \mid \gamma_{t-1}, \rst)| \leq .01 \dimx$$
and hence $\Jnl(K_t \mid \gamma_t) \leq 20 \traceb{\Pst}$. Now, we can apply \Cref{lemma:jnl_close} to conclude that for $\alpha \defeq (80 \times 20) \dimx^2 \traceb{\Pst}$ and 
\begin{align*}
	\rst \leq \min \left\{ \frac{\rnl^2}{64 \betanl \alpha^2 (1 + \opnorm{K})}, \; \frac{.01}{8 \betanl \alpha^4} \right\}
\end{align*}
it holds that $|\Jnl(K_t \mid \gamma_t , \rst) - \Jlin(K_t \mid \gamma_t)| \leq .01 \dimx $ and hence $\Jlin(K_t \mid \gamma_t) \leq 21 \traceb{\Pst}$. Therefore, $\sup_{t\geq 1} \Jlin(K_t \mid \gamma_t) \leq 21 \traceb{\Pst}$. 

Similarly, the inductive step for the random search procedure follows from noting that the exact same argument can be repeated by replacing $\Jnl(K_1 \mid \gamma_0, \rst)$ with $\Jnl(K_t \mid \gamma_{t-1})$ and $\Jlin(K_1 \mid \gamma_{0})$ with $\Jlin(K_t \mid \gamma_{t-1})$ since (by the inductive hypothesis) $\Jlin(K_{t}\mid \gamma_{t-1}) \leq 2\traceb{\Pst}$.
\subsubsection{Proof of $a)$}
 By the guarantee from Step 2, the algorithm returns a $\Khat$ which satisfies the following guarantee on the \emph{linear} system: 
\begin{align*}
	\Jlin(\Khat \mid 1) - \min_K \Jlin(K \mid 1) \leq \dimx.
\end{align*}
Therefore, $\Jlin(\Khat \mid 1) \leq 2 \traceb{\Pst}$. Now by, \Cref{lemma:exp_convergence}, the following holds,\begin{align*}
	\norm{\xtnl}^2 &\leq \opnorm{\Phat} \norm{\matx_{0}}^2 \left(1 - \frac{1}{2 \opnorm{\Phat}} \right)^t \\  
	& \leq 2\traceb{\Pst} \norm{\matx_{0}}^2 \exp \left( - \frac{t}{4 \traceb{\Pst}}\right),
\end{align*}
for $\Phat  \defeq \dlyap{A+B\Khat}{Q + \Khat^\top R \Khat}$ and all $\matx_0$ such that $\norm{\matx_0} \leq \rnl / (8 \betanl \traceb{\Pst}^2).$ 
\subsubsection{Proof of $b)$}
The bound for the number of subproblems solved by the discount annealing algorithms is similar to that of the linear case. The crux of the argument for part $b)$ is to show that any $\gamma' \in [\gamma_t, 1]$ such that 
\begin{align*}
	2.5 \Jnl(K_{t+1} \mid \gamma_t, \rst) \leq \Jnl(K_{t+1} \mid \gamma', \rst) \leq 8 \Jnl(K_{t+1} \mid \gamma_t, \rst)
\end{align*}
the following inequality must also hold: $\Jlin(K_{t+1} \mid \gamma') \geq 2 \Jlin(K_{t+1} \mid \gamma_t)$. Once we've lower bounded the cost on the \emph{linear} system, we can repeat the same argument as in Theorem 1.  Since the cost on the linear system in nondecreasing in $\gamma$, it must be the case that $\gamma'$ satisfies
\begin{align*}
	\gamma' \geq \left(\frac{1}{8 \opnorm{P_{K_{t+1}, \gamma_{t}}}^4}  + 1 \right)^2 \gamma_t \geq \left(\frac{1}{ 128 \traceb{\Pst}^4}  + 1 \right)^2 \gamma_t.
\end{align*}
Here, we have again we have used the calculation that,
\begin{align*}
\opnorm{P_{K_{t+1}, \gamma_{t}}}^4 \leq \traceb{P_{K_{t+1}, \gamma_{t}}}^4 = \Jlin(K_{t+1} \mid \gamma_t)^4 \leq (\min_K \Jlin(K\mid \gamma_t) + \dimx)^4 \leq 16 \traceb{\Pst}^4,
\end{align*}
which follows from the guarantee that (for our choice of $\rst$) policy gradients on the nonlinear system converges to a near optimal controller for the system's Jacobian linearization. Hence, as in the linear setting, we conclude that $\gamma_t \geq ( 1 / (128 \traceb{\Pst}^4)  + 1 )^{2t} \gamma_0$ and discount annealing achieves the same rate as for linear systems.

We now focus on establishing that $\Jlin(K_{t+1} \mid \gamma') \geq 2 \Jlin(K_{t+1} \mid \gamma_t)$. By guarantee from policy gradients, we have that $\Jlin(K_{t+1} \mid \gamma_t) \leq \min_K \Jlin(K \mid \gamma_t) + \dimx \leq 2\traceb{\Pst}$. Therefore, by \Cref{prop:diff_lin_nl} since
\begin{align*}
	\rst \leq \frac{.01 \rnl}{ 8 \times 2^4 \cdot \traceb{\Pst}^4\betanl} \leq \min \left\{\frac{\rnl}{2 \czero \opnorm{P_{K_{t+1}, \gamma_t}}^2} \;, \frac{.01  }{8  \czero \opnorm{P_{K_{t+1}, \gamma_t}}^4} \right\}
\end{align*}
it holds that $|\Jlin(K_{t+1} \mid \gamma_t) - \Jnl(K_{t+1} \mid \gamma_t, \rst) | \leq .01 \dimx$. 

Next, we show that $|\Jlin(K_{t+1} \mid \gamma') - \Jnl(K_{t+1} \mid \gamma', \rst) |$ is also small. In particular, the previous statement, together with the guarantee from Step 4, implies that 
\begin{align*}
 \Jnl(K_{t+1} \mid \gamma') \leq8\Jnl(K_{t+1} \mid \gamma_t) \leq 8(\Jlin(K_{t+1} \mid \gamma_t) + .01 \dimx) \leq 8.08 \Jlin(K_{t+1} \mid \gamma_t) \leq 16.16 \traceb{\Pst}.
\end{align*}
Therefore, for $\alpha \defeq (80 \times 17) \traceb{\Pst} \dimx^2$, \Cref{lemma:jnl_close} implies that if,
\begin{align*}
	r \leq \frac{\rnl^2}{64\betanl \alpha^2 (1 + \opnorm{K_{t+1}})},
\end{align*}
it holds that $|\Jnl(K_{t+1} \mid \gamma', r) - \Jlin(K_{t+1}\mid \gamma')| \leq 8 \dimx \betanl \alpha^4 r$. Hence, if $r \leq .01 / (8 \betanl \alpha^4 )$ we get that $|\Jlin(K_{t+1} \mid \gamma') -  \Jnl(K_{t+1} \mid \gamma', \rst)| \leq .01 \dimx$. Using again the fact that $\min\{ \Jlin(K \mid \gamma), \Jnl(K \mid \gamma, r)\} \geq \dimx$ for all $K,\gamma, r$ we hence conclude that 
\begin{align*}
	\Jlin(K_{t+1}\mid \gamma') &\geq .99\Jnl(K_{t+1} \mid \gamma', \rst) \\ 
		`& \geq 2.5 \times .99 \cdot \Jnl(K_{t+1} \mid \gamma_t, \rst) \\
		& \geq 2.5 \times .99^2 \cdot \Jlin(K_{t+1} \mid \gamma_t),
\end{align*}
which finished the proof of the fact that $\Jlin(K_{t+1} \mid \gamma') \geq 2 \Jlin(K_{t+1} \mid \gamma_t)$.

\subsection{Relating costs and gradients to the linear system: Proof of \Cref{prop:diff_lin_nl}}

In order to relate the properties of the nonlinear system to its Jacobian linearization, we employ the following version of the performance difference lemma. 

\begin{lemma}[Performance Difference]
\label{lemma:performance_difference}
Assume $\Jlin(K \mid \gamma) < \infty$ and define $\{\xtnl\}_{t=0}^\infty$ be the sequence of states generated according to $\xtnlp = \sqrt{\gamma}\Gnl(\xtnl, K \xtnl)$ where $\matx_{t, \mathrm{nl}} = \matx_0$. Then,
\begin{align*}
\Jnl(K \mid \gamma, \matx) - \Jlin(K \mid \gamma, \matx) = \sum_{t=0}^\infty \gamma \cdot \fnl(\xtnl, K\xtnl)^\top P_{K, \gamma} \left(\Gnl(\xtnl, K\xtnl)  + \Glin(\xtnl, K\xtnl)\right).
\end{align*}
\end{lemma}
\begin{proof}
From the definition of the relevant objectives, and \Cref{fact:value_functions}, we get that,
\begin{align}
	\Jnl(K \mid \gamma, \matx) - \Jlin(K \mid \gamma, \matx) &= \left(\sum_{t=0}^\infty \xtnl^\top(Q + K^\top R K)\xtnl \right) - \matx^\top P_{K,\gamma}\matx \notag \\
	& = \left(\sum_{t=0}^\infty \xtnl^\top(Q + K^\top R K)\xtnl \pm \xtnl^\top P_{K,\gamma} \xtnl \right) - \matx^\top P_{K,\gamma}\matx \notag \\
	& = \sum_{t=0}^\infty \xtnl^\top(Q + K^\top R K)\xtnl + \matx_{t+1, \mathrm{nl}}^\top P_{K,\gamma} \matx_{t+1, \mathrm{nl}} - \xtnl^\top P_{K,\gamma}\xtnl, \label{eq:last_step_pd}
\end{align}
where in the last line we have used the fact that $\matx_{0, \mathrm{nl}} = \matx$.  The proof then follows from the following two observations. First, by definition of state sequence, $\matx_{t,\mathrm{nl}}$,
\begin{align*}
	\matx_{t+1, \mathrm{nl}}  = \sqrt{\gamma} \cdot \Gnl(\xtnl, K\xtnl).
\end{align*}
Second, since $P_{K,\gamma}$ is the solution to a Lyapunov equation,
\begin{align*}
\xtnl^\top P_{K,\gamma}\xtnl &= \xtnl^\top (Q + K^\top R K) \xtnl +  \gamma \cdot \xtnl^\top (A+BK)^\top P_{K,\gamma} (A+BK) \xtnl \\
& = \xtnl^\top (Q + K^\top R K) \xtnl  + \gamma \cdot \Glin(\xtnl, K \xtnl)^\top P_{K,\gamma} \Glin(\xtnl, K\xtnl).
\end{align*}
Plugging these last two lines into \Cref{eq:last_step_pd}, we get that $\Jnl(K \mid \gamma, \matx) - \Jlin(K \mid \gamma, \matx)$ is equal to,
\begin{align*}
 &= \sum_{t=0}^\infty \gamma \cdot \Gnl(\xtnl, K\xtnl)^\top \Pky \Gnl(\xtnl, K\xtnl) - \gamma \cdot  \Glin(\xtnl, K\xtnl)^\top P_{K,\gamma} \Glin(\xtnl, K\xtnl) \\ 
& = \sum_{t=0}^\infty \gamma \cdot \left( \Gnl(\xtnl, K\xtnl) - \Glin(\xtnl, K\xtnl)  \right) \Pky \left( \Gnl(\xtnl, K\xtnl) + \Glin(\xtnl, K\xtnl)  \right) \\
& = \sum_{t=0}^\infty \gamma \cdot \fnl(\xtnl, K\xtnl) \Pky \left( \Gnl(\xtnl, K\xtnl) + \Glin(\xtnl, K\xtnl)  \right).
\end{align*}
\end{proof}

\subsubsection{Establishing similarity of costs}

The following lemma follows by bounding the terms appearing in the performance difference lemma.

\begin{lemma}[Similarity of Costs]
\label{lemma:similarity_of_costs}
Assume $\Jlin(K \mid \gamma) < \infty$ and define $\{\xtnl\}_{t=0}^\infty$ be the sequence of states generated according to $\xtnlp = \sqrt{\gamma}\Gnl(\xtnl, K \xtnl)$ where $\matx_{t, \mathrm{nl}} = \matx_0$. For $\matx_0$ such that,
\begin{align*}
	\matx_0^\top \Pky \matx_0 \leq \frac{\rnl^2}{4 \czero^2 \opnorm{\Pky}^3},
\end{align*}
then,
\begin{align*}
\big|\Jnl(K \mid \gamma, \matx_0) - \Jlin(K \mid \gamma, \matx_0)\big| \leq 	8 \czero \norm{\matx_0}^3 \opnorm{\Pky}^4.
\end{align*}
\end{lemma}
\begin{proof}
We begin with the following observation. Due to our assumption on $\matx_0$, we can use \Cref{lemma:exp_convergence} to conclude that for all $t \geq 0$, the following relationship holds for $V_0 \defeq \matx_{0, \mathrm{nl}}^\top P_{K,\gamma} \matx_{0, \mathrm{nl}}$, 
\begin{align}
\label{eq:exp_convergence}
	\norm{\xtnl}^2 \leq \xtnl^\top P_{K,\gamma} \xtnl \leq V_0  \cdot  \left(1- \frac{1}{2 \opnorm{P_{K,\gamma}}}\right)^t.
\end{align}
Now, from the performance difference lemma (\Cref{lemma:performance_difference}), we get that $\Jnl(K \mid \gamma, \matx) - \Jlin(K \mid \gamma, \matx)$ is equal to:
\begin{align*}
\sum_{t=0}^\infty \gamma \cdot \fnl(\xtnl, K\xtnl)^\top P_{K, \gamma} \left(\Gnl(\xtnl, K\xtnl)  + \Glin(\xtnl, K\xtnl)\right).
\end{align*}
Therefore, the difference $\big|\Jnl(K \mid \gamma, \matx_0) - \Jlin(K \mid \gamma, \matx_0)\big|$ can be bounded by,
\begin{align}
\label{eq:cs_sum}
\sum_{t=0}^\infty \underbrace{\norm{\sqrt{\gamma} \cdot \Pky^{1/2} \fnl(\xtnl, K\xtnl)}}_{\defeq T_1} \cdot  \underbrace{\norm{\sqrt{\gamma}\cdot \Pky^{1/2} \left(\Gnl(\xtnl, K\xtnl)  + \Glin(\xtnl, K\xtnl)\right)}}_{\defeq T_2}.
\end{align}
Now, we analyze each of $T_1$ and $T_2$ separately. For $T_1$, by \Cref{lemma:jacobian}, and the fact that $\gamma \leq 1$,
\begin{align*}
	\norm{\sqrt{\gamma} \cdot \Pky^{1/2} \fnl(\xtnl, K\xtnl)} &\leq \opnorm{P_{K, \gamma}}^{1/2} \norm{\fnl(\xtnl, K\xtnl)}\\
	& =  \opnorm{P_{K, \gamma}}^{1/2} \cdot \czero (\norm{\xtnl}^2 + \norm{K\xtnl}^2) \\
	& \leq \opnorm{P_{K, \gamma}}^{1/2} \czero  \cdot (\opnorm{K}^2 + 1) V_0 \cdot \left(1- \frac{1}{2 \opnorm{P_{K,\gamma}}}\right)^t,
\end{align*}
where in the last line, we have used our assumption on the initial state and \Cref{lemma:exp_convergence}. Moving onto $T_2$, we use the triangle inequality to get that
\begin{align*}
	T_2 \leq \norm{\sqrt{\gamma} \cdot P_{K, \gamma}^{1/2} \Gnl(\xtnl, K\xtnl)} + \norm{\sqrt{\gamma} \cdot P_{K, \gamma}^{1/2} \Glin(\xtnl, K\xtnl)}.
\end{align*}
For the second term above, by \Cref{lemma:dlyap} and \Cref{lemma:exp_convergence}, we have that
\begin{align*}
\norm{\sqrt{\gamma} \cdot  P_{K, \gamma}^{1/2} \Glin(\xtnl, K\xtnl)} &= \sqrt{\gamma \cdot \xtnl^\top (A+BK)^\top P_{K,\gamma} (A+BK) \xtnl} \\
& \leq \sqrt{\xtnl^\top P_{K,\gamma} \xtnl } \\
& \leq V_0^{1/2}  \left(1- \frac{1}{2 \opnorm{P_{K,\gamma}}}\right)^{t/2}.
\end{align*}
Lastly, we bound the first term by again using \Cref{lemma:exp_convergence},
\begin{align*}
\norm{\sqrt{\gamma} \cdot P_{K, \gamma}^{1/2} \Gnl(\xtnl, K\xtnl)} &= \norm{P_{K, \gamma}^{1/2} \matx_{t+1, \nl}} \leq V_0^{1/2} \left(1- \frac{1}{2 \opnorm{P_{K,\gamma}}}\right)^{\frac{t+1}{2}}.
\end{align*}
Therefore, $T_2$ is bounded by $2 V_0^{1/2} (1 - 1 / (2 \opnorm{\Pky}))^{t/ 2}$. Going back to \Cref{eq:cs_sum}, we can combine our bounds on $T_1$ and $T_2$ to conclude that,
\begin{align*}
\Jnl(K \mid \gamma, \matx) - \Jlin(K \mid \gamma, \matx) &\leq V_0^{3/2} \cdot 2 \czero  \opnorm{P_{K,\gamma}}^{1/2} (\opnorm{K}^2 + 1) \cdot \sum_{t=0}^\infty \left(1- \frac{1}{2 \opnorm{P_{K,\gamma}}}\right)^{t} \\
& = V_0^{3/2} \cdot 4 \czero  \opnorm{P_{K,\gamma}}^{3/2} (\opnorm{K}^2 + 1).
\end{align*}
Using the fact that $1 + \opnorm{K}^2 \leq 2\opnorm{\Pky}$ and $V_0 \leq \norm{\matx_0}^2 \opnorm{\Pky}$, we get that 
\begin{align*}
V_0^{3/2} \cdot 4 \czero  \opnorm{P_{K,\gamma}}^{3/2} (\opnorm{K}^2 + 1) \leq 8 \czero \norm{\matx_0}^3 \opnorm{\Pky}^4.
\end{align*}
\end{proof}

\subsubsection{Establishing similarity of gradients}

Much like the previous lemma which bounds the costs between the linear and nonlinear system via the performance difference lemma, this next lemma differentiates the performance difference lemma to bound the difference between gradients.
\begin{lemma}
\label{lemma:grad_pointwise_closeness}
 Assume $\Jlin(K \mid \gamma) <\infty$. If $\matx_0$ is such that
\begin{align*}
\matx_0 \Pky \matx_0 \leq \frac{\rnl^2}{144 \cone^2 \opnorm{\Pky}^4}
\end{align*}
then, 
\begin{align*}
\fnorm{\grad_K \Jnl(K \mid \gamma, \matx_0) - \grad_K \Jlin(K \mid \gamma, \matx_0)} \leq 48 \betanl (1 + \opnorm{B})  \opnorm{\Pky}^{7} \norm{\matx_0}^3. 
\end{align*}
\end{lemma}
\begin{proof}
Using the variational definition of the Frobenius norm,
\begin{align*}
\fnorm{\grad_K \Jnl(K \mid \gamma, \matx) - \grad_K \Jlin(K \mid \gamma, \matx)} &= \sup_{\fnorm{\Delta} \leq 1} \traceb{(\grad_K \Jnl(K \mid \gamma, \matx) - \grad_K \Jlin(K \mid \gamma, \matx))^\top \Delta} \\
& = \sup_{\fnorm{\Delta} \leq 1} \Dk  \Jnl(K \mid \gamma, \matx) -  \Dk \Jlin(K \mid \gamma, \matx),
\end{align*}
where $\Dk$ is the directional derivative operator in the direction $\Delta$. The argument follows by bounding the directional derivative appearing above. From the performance difference lemma, \Cref{lemma:performance_difference}, we have that
\begin{align*}
 \Dk  \Jnl(K \mid \gamma, \matx) -  \Dk \Jlin(K \mid \gamma, \matx)  =  \sum_{t=0}^\infty \gamma \cdot \Dk \left[ \cdot \fnl(\xtnl, K\xtnl)^\top P_{K, \gamma} \xtnlp \right].
\end{align*}
Each term appearing in the sum above can be decomposed into the following three terms,
\begin{align*}
&\underbrace{\gamma \cdot \left(\Dk  \cdot \fnl(\xtnl, K\xtnl)\right)^\top P_{K,\gamma} \xtnlp}_{\defeq T_1} \\
&\qquad+  \underbrace{ \gamma \cdot \fnl(\xtnl, K\xtnl) ( \Dk P_{K,\gamma} ) \xtnlp}_{\defeq T_2} \\
&\qquad\qquad+ \underbrace{\gamma \cdot \fnl(\xtnl, K\xtnl)  P_{K,\gamma} (\Dk \xtnlp)}_{\defeq T_3}.
\end{align*}
In order to bound each of these three terms, we start by bounding the directional derivatives appearing above. Throughout the remainder of the proof, we will make repeated use of the following inequalities which follow from \Cref{lemma:jacobian}, \Cref{lemma:exp_convergence}, and our assumption on $\matx_0$. For all $t \geq 0$,
\begin{align}
	\norm{\xtnl}^2 &\leq V_0 \left(1 - \frac{1}{2 \opnorm{\Pky}} \right)^{t} \label{eq:exp_convergence_grad}\\ 
	\norm{\xtnl} + \norm{K \xtnl} & \leq \rnl. \label{eq:can_apply_jacobian}
\end{align}
\begin{lemma}[Bounding $\Dk \xtnl$.]\label{lem:final_bound_psit} Let $\{\xtnl\}$ be the sequence of states generated according to $\xtnlp = \sqrt{\gamma} \Gnl(\xtnl, K \xtnl)$. Under the same assumptions as in \Cref{lemma:grad_pointwise_closeness}, for all $t \geq 0$
\begin{align*}
\|\Dk \xtnl|| \leq t \cdot \opnorm{\Pky}^{1/2} \left(\frac{1}{6 \opnorm{\Pky}^2} + \opnorm{B}
\right) \left(1 - \frac{1}{2\opnorm{\Pky}}\right)^{\frac{t-1}{2}} V_0^{1/2}.
\end{align*}
\end{lemma}
\begin{proof} Taking derivatives, we get that
\begin{align*}
\Dk \fnl(\xtnl, K\xtnl) = \grad_{\matx, \matu} \fnl(\matx, \matu) \eval_{(\matx, \matu) = (\xtnl, K \xtnl)} \begin{bmatrix}
	\Dk \xtnl \\
	\Dk (K\xtnl)
\end{bmatrix}.
\end{align*}
Since, $\Dk (K\xtnl) = \Delta \xtnl + K (\Dk \xtnl)$, we can rewrite the expression above as, 
\begin{equation}
\label{eq:dkfnl}
\begin{aligned}
\Dk \fnl(\xtnl, K\xtnl) &=   \grad_{\matx, \matu} \fnl(\matx, \matu) \eval_{(\matx, \matu) = (\xtnl, K \xtnl)}\begin{bmatrix}
	0 \\
	\Delta
\end{bmatrix} \xtnl \\
&\qquad+ \grad_{\matx, \matu} \fnl(\matx, \matu) \eval_{(\matx, \matu) = (\xtnl, K \xtnl)}\begin{bmatrix}
	I \\
	K
\end{bmatrix} \Dk \xtnl.
\end{aligned}
\end{equation}
Next, we compute $\Dk \xtnl$,
\begin{align*}
	\Dk \xtnl &= \sqrt{\gamma} \cdot \left( \Dk \fnl(\xtnlm, K \xtnlm) + \Dk [(A+BK)\xtnlm ] \right)\\ 
	& = \sqrt{\gamma} \cdot \left( \Dk \fnl(\xtnlm, K \xtnlm) +  B\Delta \xtnlm + (A+BK) \Dk \xtnlm \right). 
\end{align*}
Plugging in our earlier calculation for $\Dk \fnl(\xtnl, K\xtnl)$, we get that the following recursion holds for $\psi_t \defeq \Dk \xtnl$,
\begin{align}
\psi_t \defeq \Dk \xtnl &= \underbrace{\sqrt{\gamma} \cdot \left[\grad_{\matx, \matu} \fnl(\matx, \matu) \eval_{(\matx, \matu) = (\xtnlm, K \xtnlm)}\begin{bmatrix}
	0 \label{eq:def_mt}\\
	\Delta
\end{bmatrix} + B\Delta \right] \xtnlm }_{\defeq m_{t-1}} \\
& + \underbrace{\sqrt{\gamma} \cdot \left[\grad_{\matx, \matu} \fnl(\matx, \matu) \eval_{(\matx, \matu) = (\xtnlm, K \xtnlm)}\begin{bmatrix}
	I \\
	K
\end{bmatrix} +   (A+BK)\right]}_{\defeq N_{t-1}}\psi_{t-1}. \label{eq:def_Nt}
\end{align}
Using the shorthand introduced above, we can re-express the recursion as, 
\begin{align*}
\psi_t = m_{t-1} + N_{t-1} \psi_{t-1}. 
\end{align*}
Unrolling this recursion, with the base case that $\psi_0 = \Dk \matx_{0, \mathrm{nl}} = 0$, we get that 
\begin{align*}
	\psi_{t} = \sum_{j=0}^{t-1} \left(\prod_{i=j+1}^{t-1} N_{i} \right) m_{j}.
\end{align*}
Therefore, 
\begin{align}
\label{eq:psit_bound}
	\norm{\psi_t} \leq \sum_{j=0}^{t-1} \opnormm{\prod_{i=j+1}^{t-1} N_{i}} \norm{m_j}.
\end{align}
Next, we prove that each matrix $N_i$ is stable so that the product of the $N_i$ is small. By \Cref{lemma:jacobian} and our earlier inequalities,\Cref{eq:exp_convergence_grad} \& \Cref{eq:can_apply_jacobian}, we have that,
\begin{align}
	\opnormm{\grad_{\matx, \matu} \fnl(\matx, \matu) \eval_{(\matx, \matu) = (\xtnlm, K \xtnlm)}\begin{bmatrix}
	I \\
	K
\end{bmatrix}} &\leq \opnormm{\grad_{\matx, \matu} \fnl(\matx, \matu) \eval_{(\matx, \matu) = (\xtnlm, K \xtnlm)}} \opnormm{\begin{bmatrix}
	I \\
	K
\end{bmatrix}}\notag \\ 
& \leq \cone  (1 + \opnorm{K})^2 \norm{\xtnlm}  \notag \\
& \leq  \frac{1}{6 \opnorm{P_{K,\gamma}}^2} \label{eq:grad_norm_bound},
\end{align} 
where we have used our initial condition on $V_0$.
Therefore,  $N_i = \sqrt{\gamma}(A+BK) + \Delta_{N_i}$ where $\opnorm{\Delta_{N_i}} \leq 1/ (6 \opnorm{P_{K,\gamma}}^2)$ for all $i$. By definition of the operator norm,  \Cref{lemma:stability_margin}, and the fact that $\Pky \succeq I$:
\begin{align*}
\opnormm{\prod_{i=j+1}^{t-1} N_i}^2 &= \opnormm{\prnn{\prod_{i=j+1}^{t-1} N_i}^\top \prnn{\prod_{i=j+1}^{t-1} N_i }}  \\
& \leq \opnormm{\prnn{\prod_{i=j+1}^{t-1} N_i }^\top \Pky \prnn{\prod_{i=j+1}^{t-1} N_i }} \\
& \leq \opnormm{\Pky} \left(1 - \frac{1}{2 \opnorm{\Pky}}  \right)^{t - j - 1}.
\end{align*}
Then, using our bound on $\grad_{\matx,\matu}\fnl$ from \Cref{eq:grad_norm_bound}, and \Cref{lemma:exp_convergence}, we bound $m_j$ (defined in \Cref{eq:def_mt}) as,
\begin{align*}
\norm{m_j} & \leq \opnormm{\sqrt{\gamma} \left[\grad_{\matx, \matu} \fnl(\matx, \matu) \eval_{(\matx, \matu) = (\matx_{j,\mathrm{nl}}, K\matx_{j,\mathrm{nl}})}\begin{bmatrix}
	0 \\
	\Delta
\end{bmatrix} + B\Delta \right]	} \norm{\matx_{j,\mathrm{nl}}} \\ 
& \leq \left(\frac{1}{6 \opnorm{\Pky}^2} + \opnorm{B}
\right) V_0^{1/2} \left(1 - \frac{1}{2\opnorm{\Pky}}\right)^{j/2}.
\end{align*}
Returning to our earlier bound on $\norm{\psi_t}$ in \Cref{eq:psit_bound}, we conclude that,
\begin{align}
	\norm{\psi_t} & \leq \sum_{j=0}^{t-1} \opnorm{\Pky}^{1/2} \left(1 - \frac{1}{2 \opnorm{\Pky}}  \right)^{(t - j - 1) / 2} \left(\frac{1}{6 \opnorm{\Pky}^2} + \opnorm{B}
\right) V_0^{1/2} \left(1 - \frac{1}{2\opnorm{\Pky}}\right)^{j/2} \notag \\
	& \leq t \cdot \opnorm{\Pky}^{1/2} \left(\frac{1}{6 \opnorm{\Pky}^2} + \opnorm{B}
\right) \left(1 - \frac{1}{2\opnorm{\Pky}}\right)^{\frac{t-1}{2}} V_0^{1/2}, \notag 
\end{align}
concluding the proof.
\end{proof}
Using this, we can now return to bounding $\Dk \fnl(\xtnl, K\xtnl)$.
\begin{lemma}[Bounding $\Dk \fnl(\xtnl, K\xtnl)$] \label{lem:eq:final_dkfnl_bound} For all $t \ge 0$, the following bound holds:
\begin{align*}
\norm{\Dk \fnl(\xtnl, K\xtnl)} &\leq 8  \cone \opnorm{\Pky}^2 (1 + \opnorm{B})  \cdot t \left(1 - \frac{1}{2\opnorm{\Pky}}\right)^{t-1/2} \norm{\matx_0}^2. \notag
\end{align*}
\end{lemma}
\begin{proof} 
From \Cref{eq:dkfnl}, $\norm{\Dk \fnl(\xtnl, K\xtnl)}$ is less than,
\begin{align*}
 & \opnormm{\grad_{\matx, \matu} \fnl(\matx, \matu) \eval_{(\matx, \matu) = (\xtnl, K \xtnl)}} \opnormm{\begin{bmatrix}
	0 \\
	\Delta
\end{bmatrix}} \norm{\xtnl} \\
&\qquad+ \opnormm{\grad_{\matx, \matu} \fnl(\matx, \matu) \eval_{(\matx, \matu) = (\xtnl, K \xtnl)}} \opnormm{\begin{bmatrix}
	I \\
	K
\end{bmatrix}} \norm{\Dk \xtnl} .
\end{align*}
Again using the assumption on the nonlinear dynamics, we can bound the gradient terms as in \Cref{eq:grad_norm_bound}
 and get that $\norm{\Dk \fnl(\xtnl, K\xtnl)}$ is no larger than,
\begin{align*}
& \cone (1 + \opnorm{K}) \norm{\xtnl}^2 + \cone(1 + \opnorm{K})^2 \norm{\xtnl} \norm{\Dk \xtnl} \\
 \leq  & 2 \cone (1 + \opnorm{K})^2 \norm{\xtnl} \max\{ \norm{\xtnl}, \; \norm{\Dk \xtnl} \}.
\end{align*}
Seeing as how our upper bound on $\norm{\Dk \xtnl}$ in \Cref{lem:final_bound_psit} is always larger than the bound for $\norm{\xtnl}$ in \Cref{eq:exp_convergence_grad}, we can bound $\max\{ \norm{\xtnl}, \; \norm{\Dk \xtnl} \}
$ by the former. Consequently, 
\begin{align}
\norm{\Dk \fnl(\xtnl, K\xtnl)} &\leq 2 \cone (1 + \opnorm{K})^2 \opnorm{\Pky}^{1/2} V_0 \left(\frac{1}{6 \opnorm{\Pky}^2} + \opnorm{B}
\right) \cdot t \left(1 - \frac{1}{2\opnorm{\Pky}}\right)^{t-1/2} \notag \\
& \leq 8  \cone \opnorm{\Pky}^2 (1 + \opnorm{B})  \cdot t \left(1 - \frac{1}{2\opnorm{\Pky}}\right)^{t-1/2} \norm{\matx_0}^2. \notag
\end{align}
\end{proof}
Finally, we bound $\Dk \Pky$.
\begin{lemma}[Bounding $\Dk \Pky$]\label{lem:eq:final_bound_dkpky} The following bound holds:
\begin{align*}
\opnorm{\Dk \Pky} \leq 2 \opnorm{\Pky}^2(\opnorm{B} + \opnorm{K}).
\end{align*}
\end{lemma}
\begin{proof} By definition of the discrete time Lyapunov equation, 
\begin{align*}
	\Pky = Q + K^\top R K  + \gamma (A+BK)^\top \Pky (A+BK).
\end{align*}
Therefore, the directional derivative $\Dk \Pky$ satisfies another Lyapunov equation,
\begin{align*}
\Dk \Pky &= \underbrace{\Delta^\top R K + K^\top R \Delta +  \gamma \cdot (B\Delta)^\top \Pky (A+BK) +  \gamma \cdot (A+BK)^\top  \Pky B\Delta}_{\defeq E_K}  \\
&+  \gamma (A+BK)^\top (\Dk \Pky) (A+BK),
\end{align*}
implying that $\Dk \Pky = \dlyap{\sqrt{\gamma}(A+BK)}{E_K}$. By properties of the Lyapunov equation, 
\begin{align*}
\Dk \Pky &\preceq \opnorm{E_K} \dlyap{A+BK}{I} \\
&\preceq \opnorm{E_K} \dlyap{\sqrt{\gamma}(A+BK)}{K^\top RK + Q}  =  \opnorm{E_K}  \Pky.
\end{align*}
Therefore, to bound $\opnorm{\Dk \Pky}$ it suffices to bound, $\opnorm{E_k}$. Using the fact that $\opnorm{\Pky^{1/2} \sqrt{\gamma}(A + BK)} \leq \opnorm{\Pky^{1/2}}$ and that $\opnorm{\Delta} \leq 1$, a short calculation reveals that, 
\begin{align*}
	\opnorm{E_K} \leq 2 \opnorm{\Pky}(\opnorm{B} + \opnorm{K}),
\end{align*}
which together with our previous bound on $\Dk \Pky$ implies that, 
\begin{align*}
\opnorm{\Dk \Pky} \leq 2 \opnorm{\Pky}^2(\opnorm{B} + \opnorm{K}).
\end{align*}
\end{proof}
With \Cref{lem:final_bound_psit,lem:eq:final_dkfnl_bound,lem:eq:final_bound_dkpky} in place, we now return to bounding terms $T_1,T_2,T_3$. 

\paragraph{Bounding $T_1$} Recall $T_1 := \gamma \cdot \left(\Dk  \cdot \fnl(\xtnl, K\xtnl)\right)^\top P_{K,\gamma} \xtnlp$. We then have 
\begin{align*}
	\|T_1\| \leq \norm{\left( \Dk \fnl(\xtnl, K\xtnl)\right)^\top P_{K,\gamma} \xtnlp} \leq  \opnorm{\Pky} \norm{\Dk \fnl(\xtnl, K\xtnl)} \norm{\xtnlp}.
\end{align*}
Using the bound on  $\norm{\xtnlp}$ stated in \Cref{eq:exp_convergence_grad} and on $\norm{\Dk \fnl(\xtnl, K\xtnl)}$ from \Cref{lem:eq:final_dkfnl_bound}, the above simplifies to
\begin{align*}
 8 \cone \opnorm{\Pky}^{7/2} (1 + \opnorm{B})  \cdot t \left(1 - \frac{1}{2\opnorm{\Pky}}\right)^{1.5 t}\norm{\matx_0}^3 .
\end{align*}
\paragraph{Bounding $T_2$} Recall $T_2 = \fnl(\xtnl, K\xtnl)\Dk P_{K,\gamma}\xtnlp$, so that
\begin{align*}
	\norm{T_2} \leq \norm{\fnl(\xtnl, K\xtnl)} \opnorm{ \Dk P_{K,\gamma} } \norm{\xtnlp}. 
\end{align*}
Using the bound on $\norm{\fnl(\matx, \matu)}$ from \Cref{lemma:jacobian}, 
\begin{align}
	\norm{\fnl(\xtnl, K\xtnl)} & \leq \czero (1 + \opnorm{K}^2) \norm{\xtnl}^2 \notag\\
	& \leq \czero (1 + \opnorm{K}^2) \opnorm{\Pky} \left(1 - \frac{1}{2 \opnorm{\Pky}} \right)^{t} \norm{\matx_0}^2 \notag  \\ 
	& \leq 2 \czero \opnorm{\Pky}^2 \left(1 - \frac{1}{2 \opnorm{\Pky}} \right)^{t} \norm{\matx_0}^2. \label{eq:norm_fnlt_bound}
\end{align}
Therefore, using \Cref{eq:exp_convergence_grad} again and \Cref{lem:eq:final_bound_dkpky}, 
\begin{align*}
	\|T_2\| &\leq  4 \czero \opnorm{\Pky}^{9/2} (\opnorm{B} + \opnorm{K}) \left(1 - \frac{1}{2 \opnorm{\Pky}} \right)^{\frac{3}{2}t + \frac{1}{2}} \norm{\matx_0}^3  \\
	& \leq 4 \czero \opnorm{\Pky}^{5} (\opnorm{B} + 1) \left(1 - \frac{1}{2 \opnorm{\Pky}} \right)^{\frac{3}{2}t + \frac{1}{2}} \norm{\matx_0}^3.
\end{align*}
\paragraph{Bounding $T_3$.} 
Recall $T_3 := \fnl(\xtnl, K\xtnl) \P_{K,\gamma} \Dk \xtnlp$.
From \Cref{eq:norm_fnlt_bound} and \Cref{lem:final_bound_psit}, we have that
\begin{align*}
	\|T_3\| & \leq \norm{\fnl(\xtnl, K\xtnl)}  \opnorm{\P_{K,\gamma}} \opnorm{\Dk \xtnlp} \\ 
	& \leq 2 \czero \opnorm{\Pky}^4 \left(1 + \opnorm{B}
\right)  
 \cdot  (t+1) \left(1 - \frac{1}{2 \opnorm{\Pky}} \right)^{1.5t} \norm{\matx_0}^3.    
\end{align*}

\paragraph{Wrapping up} Therefore, 
\begin{align*}
\|T_1\| + \|T_2\| +\|T_3\| \leq  12 \betanl (1 + \opnorm{B})  \opnorm{\Pky}^{5}(t+1)\left(1 - \frac{1}{2 \opnorm{\Pky}} \right)^{\frac{3}{2}t} \norm{\matx_0}^3.
\end{align*}
And hence, 
\begin{align*}
\sum_{t=0}^\infty \Dk \left[ \fnl(\xtnl, K\xtnl)^\top P_{K, \gamma} \xtnlp \right] &\leq 12 \betanl (1 + \opnorm{B})  \opnorm{\Pky}^{5} \norm{\matx_0}^3 \sum_{t=0}^\infty (t+1)\left(1 - \frac{1}{2 \opnorm{\Pky}} \right)^{t}  \\ 
& = 48 \betanl (1 + \opnorm{B})  \opnorm{\Pky}^{7} \norm{\matx_0}^3.
\end{align*}
\end{proof}

\subsection{Establishing Lyapunov functions: Proof of \Cref{lemma:nonlinear_lyapunov}}
\label{subsec:proof_of_nonlinear_lyapunov}
\begin{proof}
To be concise, we use the shorthand, $\matz = (\matx, K\matx)$. We start by expanding out,
\begin{align*}
	\gamma \cdot \Gnl(\matz)^\top  P_K \Gnl(\matz) = \gamma \left( \flin(\matz)^\top P_K \flin(\matz) + \fnl(\matz)^\top P_K \fnl(\matz) + 2 \fnl(\matz)^\top P_K \flin(\matz) \right).
\end{align*}
By the AM-GM inequality for vectors, the following holds for any $\tau > 0$,
\begin{align*}
2 \fnl(\matz)^\top P_K \flin(\matz) \leq \tau \cdot \flin(\matz)^\top P_K \flin(\matz) + \frac{1}{\tau} \cdot \fnl(\matz)^\top P_K \fnl(\matz).
\end{align*}
Combining these two relationships, we get that,
\begin{align}
\label{eq:lyapunov_amgm}
\gamma \cdot \Gnl(\matz)^\top  P_K \Gnl(\matz) \leq \gamma \cdot (1 + \tau) \cdot \flin(\matz)^\top P_K \flin(\matz) + \gamma \cdot \left(1 + \frac{1}{\tau}\right) \cdot \fnl(\matz)^\top P_K \fnl(\matz).
\end{align}
Next, by properties of the Lyapunov function \Cref{lemma:dlyap}, we have that
\begin{align*}
\gamma \Glin(\matz)^\top  P_K \Glin(\matz)  = \gamma \cdot  \matx^\top (A+BK)^\top  \Pky (A +BK) \matx \leq \matx^\top \Pky \matx \left( 1 - \frac{1}{\opnorm{\Pky}}\right).
\end{align*}
Letting, $V_\matx \defeq \matx^\top \Pky \matx$, we can plug in the previous expression into \Cref{eq:lyapunov_amgm} and optimize over $\tau$ to get that,
\begin{align*}
\gamma \cdot \Gnl(\matz)^\top  \Pky \Gnl(\matz) &\leq V_\matx \left( 1 - \frac{1}{\opnorm{\Pky}}\right) + \fnl(\matz)^\top \Pky \fnl(\matz) + 2 \sqrt{V_\matx \fnl(\matz)^\top \Pky \fnl(\matz) } \\
& \leq V_\matx \left( 1 - \frac{1}{\opnorm{\Pky}}\right) + \opnorm{\Pky} \norm{\fnl(\matz)}^2 +  2 \sqrt{V_\matx \opnorm{\Pky}} \norm{\fnl(\matz)},
\end{align*}
where we have dropped a factor of $\gamma$ from the last two terms since $\gamma \leq 1$. Next, the proof follows by noting that this following inequality is satisfied,
\begin{align}
\label{eq:fnl_condition}
\opnorm{\Pky} \norm{\fnl(\matz)}^2 +  2 \sqrt{V_\matx \opnorm{\Pky}} \norm{\fnl(\matz)} \leq \frac{V_{\matx}}{2 \opnorm{\Pky}},
\end{align}
whenever,
\begin{align}
	\norm{\fnl(\matz)}& \leq \sqrt{ \frac{V_\matx}{2  \opnorm{\Pky}^2} + \frac{V_\matx}{\opnorm{\Pky}}} - \sqrt{\frac{V_\matx}{\opnorm{\Pky}}} \notag\\
	&= \sqrt{\frac{V_\matx}{\opnorm{\Pky}}}  \left( \sqrt{\frac{1}{2 \opnorm{\Pky}} + 1 } - 1\right). \label{eq:nonlinearity_cond}
\end{align}
Therefore, assuming the inequality above \Cref{eq:nonlinearity_cond}, we get our desired result showing that
\begin{align*}
	\Gnl(\matz)^\top  P_K \Gnl(\matz) \leq \matx^\top P_K \matx \cdot \left( 1 - \frac{1}{2\opnorm{P_K}} \right).
\end{align*}
We conclude the proof by showing that \Cref{eq:nonlinearity_cond} is satisfied for all $\matx$ small enough. In particular, using our bounds on $\fnl$ from \Cref{lemma:jacobian}, if $\norm{\matx} + \norm{K \matx} \leq \rnl$,
\begin{align*}
\norm{\fnl(\matz)}   \leq \czero (\norm{\matx}^2 + \norm{K\matx}^2) \leq \czero (1 + \opnorm{K}^2) \norm{\matx}^2.
\end{align*}
Since $V_\matx \geq \norm{\matx}^2$, in order for \Cref{eq:nonlinearity_cond} to hold, it suffices for $\norm{\matx}$ to satisfy
\begin{align*}
	\norm{\matx} \leq \min \left\{ \frac{\rnl}{1+ \opnorm{K}} ,\;\frac{1}{\czero \opnorm{\Pky}^{1/2}(1+\opnorm{K}^2)} \right\}.
\end{align*}
Using the fact that $\opnorm{K}^2 \leq \opnorm{\Pky}$, we can simplify upper bound on $\norm{\matx}$ to be, 
\begin{align*}
	\norm{\matx} \leq \frac{\rnl}{2 \czero \opnorm{\Pky}^{3/2}}.
\end{align*}
Note that this condition is always implied by the condition on $\matx^\top \Pky \matx$ in the statement of the proposition.
\end{proof}

\subsection{Bounding the nonlinearity: Proof of \Cref{lemma:jacobian}}
Since the origin is an equilibrium point, $\Gnl(0,0) = 0$, we can rewrite $\Gnl$ as,
\begin{align*}
	\Gnl(\matx, \matu) &= \Gnl(0,0) + \grad \Gnl(\matx,\matu) \big{|}_{(\matx, \matu) = (0,0)}
  \begin{bmatrix}
		\matx \\
		\matu
	\end{bmatrix} + \fnl(\matx, \matu), \\ 
	& =\Ajac \matx + \Bjac \matu + \fnl(\matx, \matu),
\end{align*}
for some function $\fnl$. Taking gradients,
\begin{align*}
	\grad \fnl(\matx, \matu) &= \grad_{\matx, \matu} \Gnl(\matx, \matu) \big{|}_{(\matx, \matu) = (\matx,\matu)} - \grad \Gnl(\matx,\matu) \big{|}_{(\matx, \matu) = (0,0)}. \\ 
\end{align*}
Hence, for all $\matx, \matu$ such that the smoothness assumption holds, we get that 
\begin{align*}
\opnorm{\grad_{\matx, \matu} \fnl(\matx, \matu)} \leq \betanl(\norm{\matx} + \norm{\matu}).
\end{align*}
Next, by Taylor's theorem, 
\begin{align*}
	\fnl(\matx, \matu) = \fnl(0,0) + \int_0^1 \grad \fnl(t\cdot \matx, t\cdot \matu)   \begin{bmatrix}
		\matx \\
		\matu
	\end{bmatrix} dt,
\end{align*}
and since $\fnl(0,0) = 0$, we can bound, 
\begin{align*}
	\norm{\fnl(\matx,\matu)} &\leq \int_0^1 \opnorm{\grad \fnl(t\cdot \matx, t \cdot \matu)} \norm{  \begin{bmatrix}
		\matx \\
		\matu
	\end{bmatrix}} dt \\ 
	& \leq \int_0^1 \beta t \cdot (\norm{\matx} + \norm{\matu})^2 dt. \\ 
	& \leq  2 \beta (\norm{\matx}^2 + \norm{\matu}^2) \int_0^1 t \cdot dt \\ 
	& = \beta (\norm{\matx}^2 + \norm{\matu}^2).
\end{align*}


\section{Implementing $\epseval$, $\epsgrad$, and Search Algorithms}
\label{section:finite_samples}

We conclude by discussing the implementations and relevant sample complexities of the noisy gradient and function evaluation methods, $\epsgrad$ and $\epseval$. We then use these to establish the runtime and correctness of the noisy binary and random search algorithms. 

We remark that throughout our analysis, we have assumed that $\epsgrad$ and $\epseval$ succeed with probability 1. This is purely for the sake of simplifying our presentation. As established in \cite{fazel2018global}, the relevant estimators in this section all return $\epsilon$ approximate solutions with probability $1-\delta$, where $\delta$ factors into the sample complexity polynomially. Therefore, it is easy to union bound and get a high probability guarantee. We omit these union bounds for the sake of simplifying the presentation.

\subsection{Implementing $\epseval$} 

\paragraph{Linear setting. }From the analysis in \citet{fazel2018global}, we know that if $\Jlin(K \mid \gamma) < \infty$, then
\begin{align}
\label{eq:finite_sample_approx_eval}
\Jlin(K \mid \gamma) \approx \frac{1}{\nsmp}{\sum}_{i=1}^\nsmp \Jlin^{(H)}(K \mid \gamma, \matx^{(i)}), \quad \matx^{(i)} \sim r \cdot \dxsphere,
\end{align}
where $\Jlin^{(H)}$ is the length $H$, finite horizon cost of $K$: 
\begin{align*}
	\Jlin^{(H)}(K \mid \gamma, \matx) = \sum_{j=0}^{H-1} \gamma^t \cdot \left(\matx_t^\top Q \matx_t + \matu_t^\top R \matu_t \right), \quad \matu_t = K\matx_t, \quad \matx_{t+1} = A \matx_t + B\matu_t, \quad \matx_0 = \matx.
\end{align*}
More formally, if $\Jlin(K\mid \gamma)$ is smaller than a constant $c$,  Lemma 26 from \citet{fazel2018global} states that it suffices to set $N$ and $H$ to be polynomials in $\opnorm{A}, \opnorm{B}, \Jlin(K \mid \gamma)$ and $c / \epsilon$ in order to have an $\epsilon$-approximate estimate of $\Jlin(K \mid \gamma)$ with high probability. 

On the other hand, if $\Jlin(K \mid \gamma)$ is larger than $c$ (recall that the costs should never be larger than some universal constant times $M_{\mathrm{lin}}$ during discount annealing), the following lemma proves we can detect that this is the case by setting $N$ and $H$ to be polynomials in $\opnorm{A}, \opnorm{B}$ and $c$. This argument follows from the following two lemmas:
\begin{lemma} 
\label{lemma:H_step_whp}
Fix a constant $c$, and take $H = 4c$. Then for any $K$ (possibly even with $\Jlin(K \mid \gamma) = \infty$),
\begin{align*}
\Pr_{\matx \sim \sqrt{\dimx} \cdot \dxsphere} \left[\Jlin^{(H)}(K \mid \gamma, \matx)  \geq \frac{\min\left\{\frac{1}{2} \Jlin(K \mid \gamma), c\right\}}{\dimx} \right] \geq \frac{1}{10}.
\end{align*}
\end{lemma}
Setting $c = \alpha M_{\mathrm{lin}}$ for some universal constant $\alpha$, we get that all calls to the $\epseval$ oracle during discount annealing can be implemented with polynomially many samples. Using standard arguments, we can boost this result to hold with high probability by running $N$ independent trials, where $N$ is again a polynomial in the relevant problem parameters. 

\paragraph{Nonlinear setting.} Next, we sketch why the same estimator described in \Cref{eq:finite_sample_approx_eval} also works in the nonlinear case if we replace $\Jlin^{(H)}(K\mid \gamma)$ with the analogous finite horizon cost, $\Jnl^{(H)}(K \mid \gamma, \rst)$, for the nonlinear objective. By \Cref{lemma:jnl_close} and \Cref{prop:diff_lin_nl} if the nonlinear cost is small, then costs on the nonlinear system and the linear system are \emph{pointwise} close. Therefore, previous concentration analysis for the linear setting from \cite{fazel2018global} can be easily carried over in order to implement $\epseval$ where $N$ and $H$ are depend polynomially on the relevant problem parameters. 

On the other hand if the nonlinear cost is large, then the cost on the linear system must also be large. Recall that if the linear cost was bounded by a constant, then the costs of both systems would be pointwise close by \Cref{prop:diff_lin_nl}. By \Cref{prop:finite_horizon}, we know that the $H$-step nonlinear cost is lower bounded by the cost on the linear system. Since we can always detect that the cost on the linear system is larger than a constant using polynomially many samples as per \Cref{lemma:H_step_whp}, with high probability we can also detect if the cost on the nonlinear system is large using again only polynomially many samples.

\subsection{Implementing $\epsgrad$} 
\paragraph{Linear setting.} Just like in the case of $\epseval$, \citet{fazel2018global} (Lemma 30) prove that 
\begin{align}
\label{eq:finite_sample_approx_grad}
\grad \Jlin(K \mid \gamma) \approx \frac{1}{\nsmp}{\sum}_{i=1}^\nsmp \grad \Jlin^{(H)}(K \mid \gamma, \matx^{(i)}), \quad \matx^{(i)} \sim r \cdot \dxsphere
\end{align}
where $N$ and $H$ only need to be made polynomial large in the relevant problem parameter and $1/\epsilon$ in order to get an $\epsilon$ accurate approximation of the gradient. Note that we can safely assume that $\Jlin(K\mid \gamma)$ is finite (and hence the gradient is well defined) since we can always run the test outlined in \Cref{lemma:H_step_whp} to get a high probability guarantee of this fact. In order to approximate the gradients $\grad \Jlin^{(H)}(K \mid \gamma, \matx)$, one can employ standard techniques from \cite{flaxman2005}. We refer the interested reader to Appendix D in \citet{fazel2018global} for further details.

\paragraph{Nonlinear setting. }Lastly, we remark that, as in the case of $\epseval{}$, \Cref{lemma:grad_pointwise_closeness} establishes that the gradients between the linear and nonlinear system are again pointwise close if the cost on the linear system is bounded. In the proof of \Cref{theorem:nonlinear_algorithm}, we established that during all iterations of discount annealing, the cost on the linear system during executions of policy  gradients is bounded by $M_{\mathrm{nl}}$. Therefore, the analysis from \cite{fazel2018global} can be ported over to show that $\epsilon$-approximate gradients of the nonlinear system can be computed using only polynomially many samples in the relevant problem parameters.

\subsection{Auxiliary results}

\begin{lemma}\label{lem:H_step_P} Fix any constant $c \ge 1$. Then, for 
\begin{align*}
	P_H = \sum_{j=0}^{H-1} \left((\sqrt{\gamma}(A+BK))^t \right)^\top(Q + K^\top R K)(\sqrt{\gamma}(A+BK))^t  
\end{align*}
the following relationship holds
\begin{align*}
\Jlin^{(H)}(K \mid \gamma) \ge \min\left\{\frac{1}{2} \Jlin(K \mid \gamma), c\right\}, \quad \forall H \ge 4c,
\end{align*} where $\Jlin(K \mid \gamma)$ may be infinite. 
\end{lemma}

\begin{proof} We have two cases. In the first case  $\Acl \defeq \sqrt{\gamma}(A+BK)$ has an eigenvalue $\lambda \ge 1$. Letting $\matv$ denote such a corresponding eigenvector of unit norm, one can verify that $\trace(P_H) \ge \matv^\top P_H \matv \ge \sum_{i=0}^{H-1} \lambda_i^2 \ge H$, which is at least $c$ by assumption.  

In the second case, $\Acl$ is stable, so $\Pky$ exists and its trace is $\Jlin(K \mid \gamma)$ (see \Cref{fact:value_functions}). Then, for $Q_K = Q + K^\top R K$, 
 \begin{align*}
 P_{H} = \sum_{i=0}^{H - 1} \Acl^{i \top} Q_K \Acl^{i},
 \end{align*}
 we show that if $\trace(P_H) \le c$, then $\trace(P_H) \ge \frac{1}{2}\trace(\Pky)$.

 To show this, observe that if $\trace(P_H) \le c$, then $\trace(P_H)  \le H/4$. Therefore, by the pidgeonhole principle (and the fact that $\trace(Q_K) \ge 1$), there exists some $t \in \{1,\dots,H\}$ such that $\trace(\Acl^{t \top} Q_K \Acl^{t}) \le 1/4$. Since $Q_K \succeq I$, this means that $\trace(\Acl^{t \top} \Acl^{t})= \|\Acl^t\|_{\fro}^2 \le 1/4$ as well. Therefore, letting $P_t$ denote the $t$-step value function from \Cref{def:ph}, the identity $\Pky = \sum_{n=0}^{\infty} \Acl^{nt \top} P_{t}\Acl^{nt}$ means that
\begin{align*}
\trace(\Pky) &= \trace\left(\sum_{n=0}^{\infty} \Acl^{nt \top} P_{t}\Acl^{nt}\right) \\
&\le \trace(P_t) + \|P_t\|\sum_{n \ge 1}  \|\Acl^{nt}\|_{\fro}^2\\
&\le \trace(P_t) + \|P_t\|\sum_{n \ge 1}  \|\Acl^{t}\|_{\fro}^{2n} \\
&\le \trace(P_t) + \|P_t\| \sum_{n \ge 1} (1/2)^n \\
&\le 2 \trace(P_t) \le 2\trace(P_H),
\end{align*}
where in the last line we used $t \le H$. This completes the proof.
\end{proof}

\begin{lemma}
\label{lemma:inner_product_lb}
Let $\matu, \matv \iidsim \dxsphere$ then, $\Pr\left[(\matu^\top \matv)^2 \geq 1/\dimx\right] \geq .1$.

\end{lemma}
\begin{proof}
The statement is clearly true for $\dimx=1$, therefore we focus on the case where $\dimx \geq 2$. Without loss of generality, we can take $\matu$ to be the first basis vector $\mate_1$ and let $\matv = Z / \norm{Z}$ where $Z \sim \calN(0, I)$. From these simplifications,  we observe that $(\matv^\top \matu)^2$ is equal in distribution to the following random variable,
\[
\frac{Z_1^2}{\sum_{i=1}^d  Z_i^2}, 
\]
where each $Z_i^2$ is a chi-squared random variable. Using this equivalency, we have that for arbitrary $M > 0$,
\begin{align*}
\Pr\left[(\matu^\top \matv)^2 \geq 1/\dimx\right] &=  \Pr\left[Z_1^2 \geq \frac{\sum_{i=2}^{\dimx} Z_i^2}{\dimx-1}\right] \\
& = \Pr\left[Z_1^2 \geq \frac{\sum_{i=2}^{\dimx} Z_i^2}{\dimx-1} \bigg| \sum_{i=2}^{\dimx} Z_i^2  \leq M \right] \Pr\left[ \sum_{i=2}^{\dimx}  Z_i^2  \leq M \right] \\
&+ \Pr\left[Z_1^2 \geq \frac{\sum_{i=2}^{\dimx}  Z_i^2}{\dimx-1} \bigg| \sum_{i=2}^d Z_i^2  > M \right] \Pr\left[ \sum_{i=2}^{\dimx}  Z_i^2  > M \right] \\ 
& \geq \Pr\left[Z_1^2 \geq \frac{\sum_{i=2}^{\dimx}  Z_i^2}{\dimx-1} \bigg| \sum_{i=2}^d Z_i^2  \leq M \right] \Pr\left[\sum_{i=2}^{\dimx}  Z_i^2  \leq M \right] \\
& \geq \Pr\left[Z_1^2 \geq \frac{M}{\dimx-1} \right] \Pr\left[ \sum_{i=2}^{\dimx} Z_i^2  \leq M \right].
\end{align*}
Setting $M = 2(\dimx-1)$, we get that, 
\begin{align*}
	\Pr\left[(\matu^\top \matv)^2 \geq 1/\dimx\right] \geq \Pr\left[Z_1^2 \geq 2 \right] \Pr\left[ \sum_{i=2}^{\dimx} Z_i^2  \leq 2(\dimx -1) \right].
\end{align*}
From a direct computation,
\begin{align*}
\Pr\left[Z_1^2 \geq 2 \right] \geq .15.
\end{align*} 
To bound the last term, if $Y$ is a chi-squared random variable with $k$ degrees of freedom, by Lemma 1 in \cite{laurent2000adaptive}, 
\begin{align*}
	\Pr[Y \geq k +  2\sqrt{kx} + x] \leq \exp(-x).
\end{align*}
Setting $x = 2\sqrt{2(\dimx-1)k} + 2\dimx + k -2$ we get that $k + 2\sqrt{kx}+ x = 2(\dimx-1)$. Substituting in $k =\dimx-1$, we conclude that, 
\begin{align*}
\Pr\left[ \sum_{i=2}^{\dimx} Z_i^2  \leq 2(\dimx -1) \right] \geq 1 - \exp(-2\sqrt{2}(\dimx -1) - 3 \dimx + 3),
\end{align*}
which is greater than .99 for $\dimx \geq 2$. 
\end{proof}

\newtheorem*{lemma:H_step_whp_restated}{\Cref{lemma:H_step_whp} (restated)}

\begin{lemma:H_step_whp_restated} 
Fix a constant $c$, and take $H = 4c$. Then for any $K$ (possibly even with $\Jlin(K \mid \gamma) = \infty$),
\begin{align*}
\Pr_{\matx \sim \sqrt{\dimx} \cdot \dxsphere} \left[\Jlin^{(H)}(K \mid \gamma, \matx)  \geq \frac{\min\left\{\frac{1}{2} \Jlin(K \mid \gamma), c\right\}}{\dimx} \right] \geq \frac{1}{10}.
\end{align*}
\end{lemma:H_step_whp_restated}
\begin{proof} Observe that for the finite horizon value matrix $P_H = \sum_{i=0}^{H-1} (A+BK)^{i\top}(Q + K^\top R K) (A+BK)^i$, we have $\Jlin^{(H)}(K \mid \gamma, \matx) = \matx^\top P_H \matx$. We now observe that, since $P_H \succeq 0$, it has a (possibly non-unique) top eigenvector $v_1$ for which
\begin{align*}
\Jlin^{(H)}(K \mid \gamma, \matx) = \matx^\top P_H \matx &\ge  \langle v_1, \frac{\matx}{\sqrt{\dimx}} \rangle^2  \cdot \dimx\|P_H\| \\
&\ge  \langle v_1, \frac{\matx}{\sqrt{\dimx}} \rangle^2 \cdot \trace(P_H) \\
&= \underbrace{\langle v_1, \frac{\matx}{\sqrt{\dimx}}}_{:=Z} \rangle^2\cdot \Jlin^{(H)}(K \mid \gamma)
\end{align*}
Since $\matx/\sqrt{\dimx}  \sim \dxsphere$, \Cref{lemma:inner_product_lb} ensures that $\Pr[Z \ge 1/\dimx] \ge 1/10$. Hence,
\begin{align*}
\Pr_{\matx \sim \sqrt{\dimx} \cdot \dxsphere} \left[\Jlin^{(H)}(K \mid \gamma, \matx)  \geq \frac{ \Jlin^{(H)}(K \mid \gamma)}{\dimx} \right] \geq \frac{1}{10}.
\end{align*}
The bound now follows from invoking \Cref{lem:H_step_P} to lower bound $\Jlin^{(H)}(K \mid \gamma, \matx)  \ge \min\left\{\frac{1}{2} \Jlin(K \mid \gamma), c\right\}$ provided $H \ge 4c$.

\end{proof}

In the following lemmas, we define $P_H$ to be the following matrix where $K$ is any state-feedback controller.  
\begin{align}
\label{def:ph}
	P_H = \sum_{j=0}^{H-1} \left((\sqrt{\gamma}(A+BK))^t \right)^\top(Q + K^\top R K)(\sqrt{\gamma}(A+BK))^t  
\end{align}
Similarly, we let $\Jnl^{(H)}(K \mid \gamma , \matx_0)$ be the horizon $H$ cost of the nonlinear dynamical system:
\begin{align}
&\Jnl^{(H)}(K  \mid \matx, \gamma) \defeq  \sum_{t=0}^{H-1}  \matx_t^\top Q \matx_t + \matu_t^\top R \matu_t \label{eq:nonlinear_cost} \\
&\text{ s.t }  \matu_t = K\matx_t,\quad   \matx_{t+1} = \sqrt{\gamma} \cdot \Gnl( \matx_t, \matu_t), \quad \matx_0 = \matx. \label{eq:nonlinear_dynamics}
\end{align}
And again overloading notation like before, we let $\Jnl^{(H)}(K \mid \gamma, r)  \defeq \E_{\matx \sim r \cdot \dxsphere} \left[ \Jnl^{(H)}(K \mid \gamma,  \matx) \right] \times \frac{\dimx}{r^2}$.

\begin{lemma}\label{lem:nl_lb} Fix a horizon $H$, constant $\alpha \in (0,\dimx)$, $\matx_0 \in \R^{\dimx}$, and suppose that 
\begin{align*}
\|\matx_0\|^2 \cdot \frac{\alpha }{2H^2 \betanl^2 (1+\|K\|^2) \dimx} \le \rnl^4.
\end{align*}
Furthermore, define $\calX_\alpha := \{\matx_0 \in \R^{\dimx}:\langle \matx_0, \matv_{\max}(P_H)\rangle^2 \ge \alpha\|\matx_0\|^2 / \dimx\}$. Then, if $\matx_0 \in \calX_{\alpha}$, it holds that 
\begin{align*}
\Jnl^{(H)}(K \mid \gamma, \matx_0) \ge   \min\left\{\frac{\alpha\|\matx_0\|^2 }{4\dimx^2}\cdot\Jlin^{(H)}(K \mid \matx_0,\gamma), \sqrt{\frac{\alpha}{\dimx}} \frac{\|\matx_0\|}{2H\betanl (1+\|K\|)}\right\}.\end{align*}
\end{lemma}

\begin{proof} Fix a constant $r_1 \le \rnl$ to be selected. Throughout, we use the shorthand $\beta_1^2 = \betanl^2 (1+\|K\|^2)$ We consider two cases.  
\paragraph{Case 1:} The initial point $\matx_0$ is such that it is always the case that $\|\matx_t\| \le r_{1}$ for all $t \in \{0,1,\dots,H-1\}$ and $ \Jnl^{(H)}(K  \mid \gamma, \matx_0) \le r_1^2$. 
Observe that we can write the nonlinear dynamics as
\begin{align*}
\matx_{t+1} = \sqrt{\gamma}(A+BK) \matx_t + \matw_t,
\end{align*}
where $\matw_t = \sqrt{\gamma}\cdot \fnl(\matx_t,K\matx_t)$. We now write:
\begin{align*}
\matx_{t} &= \matx_{t;0} + \matx_{t;w}, \text{where }\\
&\qquad \matx_{t;0} = \sqrt{\gamma}(A+BK)^{t}\matx_0\\
&\qquad \matx_{t;w} = \sum_{i=0}^{t-1}\gamma^{i/2}(A+BK)^{i}\matw_{t-i}.
\end{align*}
Then, setting $Q_K = Q + K^\top R K$,
\begin{align*}
\Jnl^{(H)}(K \mid \gamma, \matx_0) &= \sum_{t=0}^{H-1} \matx_t^\top Q_K \matx_t\\
&= \sum_{t=0}^{H-1} \left(\matx_{t;0}^\top Q_K \matx_{t;0} + 2\matx_{t;0}^\top Q_K \matx_{t;w}  + \matx_{t;w}^\top Q_K \matx_{t;w} \right)\\
&\ge \frac{1}{2}\sum_{t=0}^{H-1} \matx_{t;0}^\top Q_K \matx_{t;0} - \sum_{t=0}^{H-1} \matx_{t;w}^\top Q_K \matx_{t;w}\\
&= \frac{1}{2}\matx_0^\top P_H \matx_0 - \sum_{t=0}^{H-1} \matx_{t;w}^\top Q_K \matx_{t;w},
\end{align*}
where (a) the last inequality uses the elementary inequality $\langle \matv_1, \Sigma \matv_1 \rangle + \langle \matv_2, \Sigma \matv_2 \rangle+ 2 \langle \matv_1, \Sigma \matv_2 \rangle \ge \frac{1}{2}\langle \matv_1, \Sigma \matv_1 \rangle - \langle \matv_2, \Sigma \matv_2 \rangle$ for any pair of vectors $\matv_1, \matv_2$ and $\Sigma \succeq 0$, and (b) the last inequality recognizes how
\begin{align*}
	\sum_{t=0}^{H-1} \matx_{t;0}^\top Q_K \matx_{t;0} = \Jlin^{(H)}(K \mid \matx_0,\gamma) = \matx_0^\top P_H \matx_0
\end{align*}
for $P_H$ defined above in \Cref{def:ph}. Moreover, for any $t$,
\begin{align*}
\sum_{t=0}^{H-1}\matx_{t;w}^\top Q_K \matx_{t;w} &= \sum_{t=0}^{H-1}\sum_{i=0}^{t-1}\sum_{j=0}^{t-1} \matw_{t-i}^\top\gamma^{(i+j)/2}((A+BK)^{i})^\top Q_K(A+BK)^{j}\matw_{t-j}\\
&\le  H\sum_{t=0}^{H-1}\sum_{i=0}^{t-1}\matw_{t-i}^\top \gamma^i ((A+BK)^{i})^\top Q_K(A+BK)^{i}\matw_{t-i}\\
&=  H\sum_{t=0}^{H-2}\matw_{t}^\top\left(\sum_{i=0}^{H-t-1}\gamma^i ((A+BK)^{i})^\top Q_K(A+BK)^{i}\right)\matw_{t}\\
&=  H\sum_{t=0}^{H-2}\matw_{t}^\top P_{H-t-1}\matw_{t} \le H^2 \|P_H\|_{\op} \max_{t\in\{0,\dots,H-2\}}\|\matw_t\|^2.
\end{align*}
Now, because $\|\matx_t\| \le r_1 \le \rnl$, \Cref{lemma:jacobian} lets use bound $\|\matw_t\|^2\le \betanl^2 (1+\|K\|^2) \|\matx_t\|^4 \le \beta_1^2 r_1^4$, where we adopt the previously defined shorthand $\beta_1^2 = \betanl^2 (1+\|K\|^2)$. Therefore, 
\begin{align*}
\Jnl^{(H)}(K \mid \gamma, \matx_0) \ge \frac{1}{2}\matx_0^\top P_H \matx_0 - H^2 \beta_1^2 r_1^4 \|P_H\|_{\op}. 
\end{align*}
Next, if $\matx_0 \in \calX_{\alpha}$,
\begin{align*}
\Jnl^{(H)}(K \mid \gamma, \matx_0) &\ge \frac{\alpha }{2\dimx} \|\matx_0\|^2 \|P_H\|  - H^2 \beta_1^2 r_1^4 \|P_H\|_{\op}\\
&= \|P_H\|\left(\frac{\alpha }{2\dimx} \|\matx_0\|^2   - H^2 \beta_1^2 r_1^4\right).
\end{align*}
In particular, selecting $r_1^4 = \frac{\alpha }{4 \beta_1^2  \dimx H^2} \|\matx_0\|^2$ (which ensures $r_1 \le r_{\star}$ by the conditions of the lemma), it holds that,
\begin{align*}
\Jnl^{(H)}(K \mid \gamma, \matx_0) \ge \frac{\|P_H\|\alpha }{4\dimx} \|\matx_0\|^2 \ge \frac{\trace(P_H)\alpha }{4\dimx^2} \|\matx_0\|^2 = \frac{\alpha\Jlin^{(H)}(K,\gamma) }{4\dimx^2} \|\matx_0\|^2.
\end{align*}

\paragraph{Case 2:} The initial point $\matx_0$ is such that it is always the case that either $\|\matx_t\| \ge r_{1}$ for all $t \in \{0,1,\dots,H-1\}$ or $ \Jnl^{(H)}(K \mid \gamma, \matx_0) \ge r_1^2$. Therefore, in either case, $\Jnl^{(H)}(K \mid \gamma, \matx_0) \ge r_1^2$. For our choice of $r_1$, this gives
\begin{align*}
\Jnl^{(H)}(K \mid \gamma, \matx_0) \ge \sqrt{\frac{\alpha  \|x_0\| ^2}{4\dimx H^2\beta_1^2}} = \sqrt{\frac{\alpha  \|x_0\|^2}{4\dimx H^2\betanl^2(1+\|K\|^2)}} \ge \sqrt{\frac{\alpha}{\dimx}} \frac{\|x_0\|}{2H\betanl (1+\|K\|)}.
\end{align*}
 Combining the cases, we have
\begin{align*}
\Jnl^{(H)}(K \mid \gamma, \matx_0) \ge   \min\left\{\frac{\alpha\|\matx_0\|^2 }{4\dimx^2}\cdot\Jlin^{(H)}(K \mid \matx_0,\gamma), \sqrt{\frac{\alpha}{\dimx}} \frac{\|x_0\|}{2H\betanl (1+\|K\|)}.\right\}.
\end{align*}

\end{proof}

\begin{proposition}
\label{prop:finite_horizon}
 Let $c$ be a given (integer) tolerance $c$ and  $\calX_{\alpha} := \{\matx \in \R^{d}: \langle \matx , \matv_{\max}(P_H)\rangle^2 \ge \alpha\|\matx\|^2/ \dimx\}$ be defined as in \Cref{lem:nl_lb}. Then, for $\alpha \in (0,2]$, $H = 4c$, and $\matx_0 \in \calX_{\alpha}$ satisfying,
\begin{align}
\|\matx_0\|^2  \le c^2\rnl^4, \quad \text{and} \quad \|\matx_0\|\le\frac{\dimx}{64c^2 \betanl (1+\|K\|)},
\end{align}
it holds that:
\begin{align*}
\frac{ \Jnl^{(H)}(K \mid \gamma, \matx_0)}{\|\matx_0\|^2} \ge \frac{\alpha}{8\dimx^2}\min\left\{\Jlin(K\mid \gamma),c\right\}.
\end{align*}
Moreover, for $r \leq \min\{c\rnl^2,  \frac{\dimx}{64c^2 \betanl (1+\|K\|)}  \}$,
\begin{align*}
\Jnl^{(H)}(K \mid \gamma, r) \ge \frac{1}{80\dimx^2} \min\left\{\Jlin(K \mid \gamma),c\right\}.
\end{align*}
\end{proposition}
\begin{proof}
 From \Cref{lem:nl_lb}, it holds that for $H = 4c$,
\begin{align}
\|\matx_0\|^2 \cdot \frac{\alpha }{32 c^2 \betanl^2 (1+\|K\|^2) \dimx} \le \rnl^4 \label{eq:xnot_rnl_condition}
\end{align}
and hence 
\begin{align*}
\Jnl^{(H)}(K \mid \gamma, \matx_0) \ge   \min\left\{\frac{\alpha\|\matx_0\|^2 }{4\dimx^2}\cdot\Jlin^{(H)}(K \mid \matx_0,\gamma), \sqrt{\frac{\alpha}{\dimx}} \frac{\|\matx_0\|}{8c\betanl (1+\|K\|)}\right\}.
\end{align*}
Note that, due to $\alpha / \dimx \le 1$, $\betanl^2 \ge 1$,  \Cref{eq:xnot_rnl_condition}  holds as soon as $\|\matx_0\| \le c^2 \rnl^4$. By \Cref{lem:H_step_P}, it then follows that:
\begin{align*}
\Jnl^{(H)}(K \mid \gamma, \matx_0) &\ge   \min\left\{\frac{\alpha\|\matx_0\|^2 }{4\dimx^2}\cdot \min\left\{\frac{1}{2}\Jlin(K \mid \gamma), c\right\}, \sqrt{\frac{\alpha}{\dimx}} \frac{\|\matx_0\|}{8c\betanl (1+\|K\|)}\right\}\\
&\ge   \min\left\{\frac{\alpha\|\matx_0\|^2 }{8\dimx^2}\min\left\{\Jlin(K \mid \gamma), c\right\}, \sqrt{\frac{\alpha}{\dimx}} \frac{\|\matx_0\|}{8c\betanl (1+\|K\|)}\right\}\\
&=   \frac{\alpha\|\matx_0\|^2 }{8\dimx^2}\min\left\{\Jlin(K \mid \gamma), \,c,\, \left(\frac{\alpha\|\matx_0\|^2 }{8\dimx^2}\right)^{-1}\sqrt{\frac{\alpha}{\dimx}} \frac{\|\matx_0\|}{8c\betanl (1+\|K\|)}\right\}.
\end{align*}
Simplifying, the last term gives 
\begin{align*}
\left(\frac{\alpha\|\matx_0\|^2 }{8\dimx^2}\right)^{-1}\sqrt{\frac{\alpha}{\dimx}} \frac{\|\matx_0\|}{8c\betanl (1+\|K\|)} = \frac{1}{c\|\matx_0\|} \cdot \left(\frac{\dimx^3}{\alpha}\right)^{1/2} \frac{1}{64 \betanl (1+\|K\|)} \ge \frac{1}{c\|\matx_0\|} \cdot \frac{\dimx }{64 \betanl (1+\|K\|)},
\end{align*}
where we use $\dimx/\alpha \ge 1$. Thus, for 
\begin{align*}
\|\matx_0\| \le  \frac{\dimx}{64c^2 \betanl (1+\|K\|)},
\end{align*}
the third term in the minimum is at most $c$, so that
\begin{align*}
\Jnl^{(H)}(K \mid \gamma, \matx_0)  \ge  \frac{\alpha\|\matx_0\|^2 }{8\dimx^2}\min\left\{\Jlin(K \mid \gamma), \,c\right\}.
\end{align*}
Lastly, using \Cref{lemma:inner_product_lb},
\begin{align*}
\Jnl^{(H)}(K \mid \gamma, r) &= \frac{\dimx}{r^2}\cdot \E_{\matx_0 \sim r \cdot \dxsphere}\Jnl^{(H)}(K \mid \gamma, \matx_0) \\
&\ge  \Pr[x_0 \in \calX_1] \cdot \frac{1}{8\dimx^2}\min\left\{\Jlin(K \mid \gamma),c\right\}\\
&\ge  \frac{1}{80\dimx^2}\min\left\{\Jlin(K \mid \gamma),c\right\}.
\end{align*}
\end{proof}

\subsection{Search analysis}
\label{subsec:search}
\begin{figure}[t!]
\setlength{\fboxsep}{2mm}
\begin{boxedminipage}{.48\textwidth}
\vspace{5pt}
\begin{center}
\underline{Noisy Binary Search} \\
\end{center}
{\bf Require:} $\fbar_1$ and $\fbar_2$ as defined in \Cref{lemma:binary_search}.\\ 

{\bf Initialize:} $b_0\leftarrow 0, u_0 \leftarrow 1$,  $c \geq \fbar_2 + 2\alpha$\\ \hspace{90pt}\quad for $\alpha \defeq \alpha =  \min \{ |f_1 - \fbar_1|, |f_2 - \fbar_2| \}$, $\epsilon \in (0, \alpha / 2)$\\

{\bf For $t = 1, \dots$}
\begin{center}
\begin{enumerate}
	\item Query $a_t \leftarrow \epseval(f, x_t, c)$ where $$x_t = \frac{b_t+u_t}{2}$$
	\item If $a_t > \fbar_2 + \epsilon $, update $u_t \leftarrow x_t$ and $b_t \leftarrow b_{t-1}$

	\item Else if,  $a_t < \fbar_1 + \epsilon$, update $b_t \leftarrow x_t$ and $u_t \leftarrow u_{t-1}$
	\item Else, break and return $x_t$

\end{enumerate}
\end{center}
\vspace{5pt}
\end{boxedminipage}
\begin{boxedminipage}{.48\textwidth}
\vspace{5pt}
\begin{center}
\underline{Noisy Random Search} \\
\end{center}
{\bf Require:} $\fbar_1$, $\fbar_2$ as defined in \Cref{lemma:random_search}.\\ 

{\bf Initialize:} $c \geq \fbar_2 + 2\alpha$\\ for $\alpha \defeq  \min \{ |f_1 - \fbar_1|, |f_2 - \fbar_2| \}$, $\epsilon \in (0, \alpha / 2)$\\ 

{\bf For $t = 1, \dots$}
\begin{center}
\begin{enumerate}
	\item Sample $x$ uniformly at random from $[0,1]$
	\item Query $a \leftarrow \epseval(f, x, c)$ 

	\item If $a \in [\fbar_1, \fbar_2]$, break and return $x$, 
\end{enumerate}
\end{center}
\vspace{75pt}
\end{boxedminipage}

\caption{Algorithms for 1 dimensional search used as subroutines for Step 4 in the discount annealing algorithm.}
\label{fig:noisy_binary_search}
\end{figure}

\begin{lemma}
\label{lemma:binary_search}
Let $f:[0,1] \rightarrow \R \cup \{\infty \}$ be a nondecreasing function over the unit interval. Then, given $\fbar_1, \fbar_2$ such that $[\fbar_1, \fbar_2] \subseteq [f_1, f_2]$ and for which there exist $[x_1, x_2] \subseteq [0,1]$ such that for all $x' \in [x_1, x_2]$, $f(x') \in [\fbar_1, \fbar_2]$,  binary search as defined in \Cref{fig:noisy_binary_search} returns a value $x_\star \in [0,1]$ such that $f(x_\star) \in [f_1, f_2]$ in at most $\ceil{\log_2(1 / \Delta)}$ many iterations where $\Delta = x_2 - x_1$.
\end{lemma}
\begin{lemma}
\label{lemma:random_search}
Let $f:[0,1] \rightarrow \R \cup \{\infty\}$ be a function over the unit interval. Then, given given $\fbar_1, \fbar_2$ such that $[\fbar_1, \fbar_2] \subseteq [f_1, f_2]$ and for which there exist $[x_1, x_2] \subseteq [0,1]$ such that for all $x' \in [x_1, x_2]$, $f(x') \in [\fbar_1, \fbar_2]$, with probability $1-\delta$, noisy random search as defined in \Cref{fig:noisy_binary_search}  returns a value $x_\star \in [0,1]$ such that $f(x_\star) \in [f_1, f_2]$ in at most $1 / \Delta \log(1 / \delta)$ many iterations where $\Delta = x_2 - x_1$.
\end{lemma}

The analysis of the correctness and runtime of this classical algorithms for search problems in 1 dimension are standard. We omit the proofs for the sake of concision.

\section{Additional Experiments}\label{sec:experiments_app}

\subsection{Cart-pole dynamics}\label{sec:cartpole}

The state of the cart-pole is given by $\matx=(x,\theta,\dot{x},\dot{\theta})$, where $x$ and $\dot{x}$ denote the horizontal position and velocity of the cart, and $\theta$ and $\dot{\theta}$ denote the angular position and velocity of the pole, with $\theta=0$ corresponding to the upright equilibrium. 
The control input $u \in\R$ corresponds to the horizontal force applied to the cart. 
The continuous time dynamics are given by 
\begin{equation*}
\left[\begin{array}{cc}
m_p + m_c  & -m_p l \cos(\theta) \\ 
-m_p l \cos(\theta) & m_pl^2
\end{array}\right]
\left[\begin{array}{c}
\ddot{x} \\ \ddot{\theta}
\end{array}\right]
=
\left[\begin{array}{c}
u - m_p l \sin(\theta) \dot{\theta}^2 \\ 
m_p g l \sin(\theta)
\end{array}\right]
\end{equation*}
where $m_p$ denotes the mass of the pole, $m_c$ the mass of the cart, $l$ the length of the pendulum, and $g$ acceleration due to gravity. 
For our experiments, we set all parameters to unity, i.e. $m_p=m_c=l=g=1$.

\subsection{$\hinf$ synthesis}\label{sec:hinf}
Consider the following linear system
\begin{equation}\label{eq:hinfsys}
\matx_{t+1} = A\matx_t + B \matu_t + \matw_t, \quad 
\matz_t = \left[\begin{array}{c}
Q^{1/2} \matx_t \\ R^{1/2}\matu_t
\end{array}\right],
\end{equation}
where $\matw$ denotes an additive disturbance to the state transition, and $\matz$ denotes the so-called performance output.
Notice that $\|\matz_t\|_2^2 = \matx_t^\top Q \matx_t + \matu_t^\top R\matu_t$.
The $\hinf$ optimal controller minimizes the $\hinf$-norm of the closed-loop system from input $\matw$ to output $\matz$, 
i.e. the smallest $\eta$ such that 
\begin{equation*}
\sum_{t=0}^T \|\matz_t\|_2^2 \leq \eta \sum_{t=0}^T \eta \| \matw_t\|_2^2
\end{equation*}
holds for all $\matw_{0:T}$, all $T$, and $\matx_0=0$. 
In essence, the $\hinf$ optimal controller minimizes the effect of the worst-case disturbance on the cost.
In this setting, additive disturbances $\matw$ serve as a crude (and unstructured) proxy for modeling error. 
We synthesize the $\hinf$ optimal controller using Matlab's \texttt{hinfsyn} function. 
For the system \eqref{eq:hinfsys} with perfect state observation, we require only static state feedback $\matu_t=K\matx_t$ to implement the controller.

\subsection{Difference between LQR and discount annealing as a function of discount}

In \Cref{fig:err_vs_discount} 
we plot the error $\| K_\text{pg}^*(\gamma) - K_\text{lin}^*(\gamma)\|_F$ between the policy $K_\text{pg}^*(\gamma)$ returned by policy gradient and the optimal policy $K_\text{lin}^*(\gamma)$ for the (damped) linearized system, as a function of the discount $\gamma$ during the discount annealing process. 
Observe that for small radius of the ball of initial conditions ($r=0.05$), the optimal controller from policy gradient remains very close to the exact optimal controller for the linearized system; however, for larger radius ($r=0.7$) the controller from policy gradient differs significantly.

\begin{figure}[b!]
	\centering
	\includegraphics[width=0.5\textwidth]{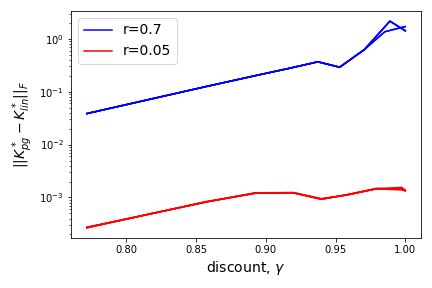}
	\caption{Error between the policy returned by policy gradient and optimal LQR policy for the (damped) linearized system during discount annealing, for two different radius values, $r=0.05$ and $r=0.7$, of the ball of initial conditions. Five independent trials are plotted for each radius value.The values across trials highly overlap.}
	\label{fig:err_vs_discount}
\end{figure}

\end{document}